\documentclass[reprint,
superscriptaddress,
amsmath,
amssymb,
aps]{revtex4-2}

\usepackage[english]{babel}
\usepackage[usenames]{color}
\usepackage[document]{ragged2e}

\usepackage{amsfonts}
\usepackage{amsmath}
\usepackage{amssymb}
\usepackage{bm}
\usepackage{braket}
\usepackage{caption}
\usepackage{dcolumn}
\usepackage{enumitem}
\usepackage{float}
\usepackage[final]{graphicx}
\usepackage{hyperref}
\usepackage{lineno}
\usepackage{mathtools}
\usepackage{orcidlink}
\usepackage{soul}
\usepackage{subfigure}
\usepackage{url}
\usepackage{verbatim}
\usepackage{xfrac}
\usepackage{xspace}
\usepackage{amsthm}
\usepackage{hhline}
\renewcommand{\arraystretch}{1.2}
\usepackage{algorithm}
\usepackage{algpseudocode}

\DeclarePairedDelimiter{\ceil}{\lceil}{\rceil}

\newtheorem{theorem}{Theorem}
\newtheorem{lemma}{Lemma}

\newtheorem{corollary}{Corollary}

\algrenewcommand\algorithmicrequire{\textbf{Input:}}
\algrenewcommand\algorithmicensure{\textbf{Output:}}

\begin{document}

\title{Resource-Scalable Fully Quantum Metropolis--Hastings\\
for Integer Linear Programming}

\author{Gabriel Escrig}
\email{gescrig@ucm.es}
\affiliation{Departamento de Física Teórica, Universidad Complutense de Madrid.}

\author{Roberto Campos}
\email{robecamp@ucm.es}
\affiliation{Departamento de Física Teórica, Universidad Complutense de Madrid.}

\author{M. A. Martín-Delgado}
\email{mardel@ucm.es}
\affiliation{Departamento de Física Teórica, Universidad Complutense de Madrid.}
\affiliation{CCS-Center for Computational Simulation, Universidad Politécnica de Madrid.}

\justifying

\begin{abstract}

Integer linear programming (ILP) remains computationally challenging due to its NP-complete nature despite its central role in scheduling, logistics, and design optimization. We introduce a fully quantum Metropolis--Hastings algorithm for ILP that implements a coherent random walk over the discrete feasible region using only reversible quantum circuits, without quantum-RAM assumptions or classical pre/post-processing. Each walk step is a unitary update that prepares coherent candidate moves, evaluates the objective and constraints reversibly---including a constraint-satisfaction counter to enforce feasibility---and encodes Metropolis acceptance amplitudes via a low-overhead linearized rule. At the logical level, the construction uses $\mathcal{O}(n\log_2 N)$ qubits to represent $n$ integer variables over the interval $[-N,\,N-1]$, and the Toffoli-equivalent cost per Metropolis step grows linearly with the total logical qubit count. Using explicit ripple-carry adder constructions, we support linear objectives and mixed equality/inequality constraints. Numerical circuit-level simulations on a broad ensemble of randomly generated instances validate the predicted linear resource scaling and exhibit progressive thermalization toward low-cost feasible solutions under the annealing schedule. Overall, the method provides a coherent, resource-characterized baseline for fully quantum constraint programming and a foundation for incorporating additional quantum speedups in combinatorial optimization.
\end{abstract}

\maketitle
\justifying

\section{Introduction}

\begin{figure*}[t]
  \centering
  \includegraphics[width=\textwidth]{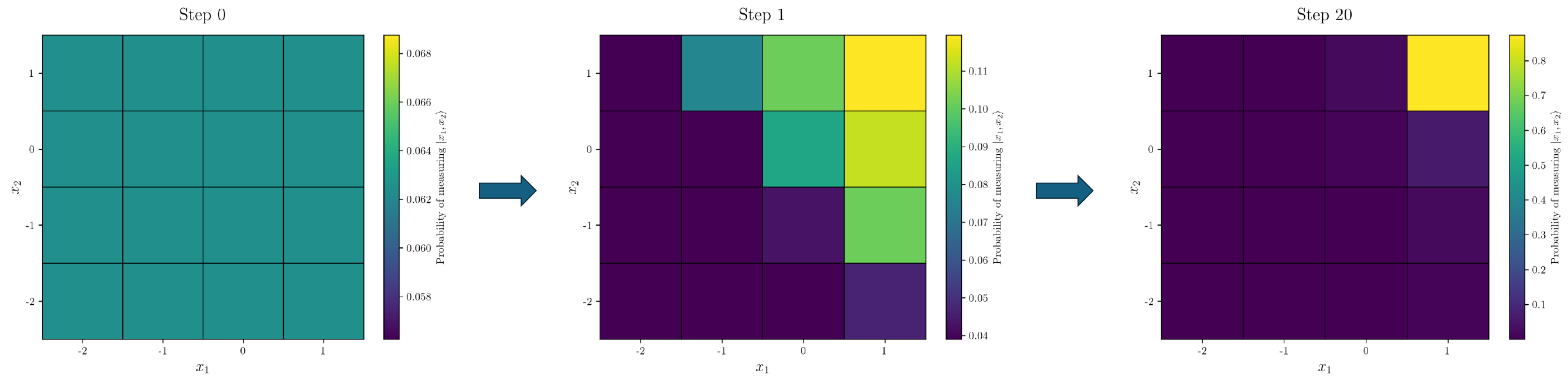}
  \caption{\justifying Heat map of the probability distribution obtained after applying the \textit{Randomized Quantum Metropolis} algorithm (see Sec.~\ref{sec:overview}) for $20$ steps. Color intensity encodes the probability of measuring the quantum state $\ket{x_1, x_2}$, where each grid cell corresponds to a discrete point in the feasible region defined by the constraint $x_1 + x_2 \ge 0$. The optimization function is given by $f(x_1, x_2) = -2x_1 - x_2$, and the distribution reflects the thermalized process. The problem is solved for 2 qubits of discretization per variable, leaving the possible search space as the square $[-2,1]\times[-2,1]$.}
  \label{fig:3_step_inference}
\end{figure*}

Integer linear programming sits at the heart of operations research, underpinning critical applications in scheduling, logistics, and design. Defined by the optimization of a linear objective function subject to linear constraints over integer variables, ILP provides a modeling framework broad enough to cover a vast range of industrial problems \cite{clautiaux2025, zhang2023}. While the Simplex algorithm remains the de‑facto standard for continuous relaxations \cite{koch2022}, the integrality constraints render the general problem NP-complete \cite{Schrijver1986}. Consequently, modern commercial solvers rely on sophisticated branch‑and‑bound or branch‑and‑cut frameworks, augmented by primal heuristics, to certify optimality \cite{Gomory1958, Land1960, Padberg1991}.

In contexts where diverse high-quality solutions are preferred over a single optimal proof—such as generating warm starts for exact solvers—Metropolis--Hastings (MH) sampling serves as a powerful primitive \cite{fischetti2003}. As a lightweight Markov-chain Monte Carlo method, MH enables the exploration of large discrete state spaces and admits efficient parallelization \cite{calderhead2014}. However, classical ILP solvers face three inherent barriers: combinatorial explosion of the search space, scalability bottlenecks related to constraint evaluation, and reliance on structural heuristics that may not generalize.

From a complexity‑theoretic standpoint, ILP’s NP‑completeness precludes universal polynomial‑time algorithms—classical or quantum. While quantum speed‑ups have been reported for semidefinite programming, they typically rely on sparsity and quantum‑RAM (QRAM) data‑access assumptions \cite{kerenidis2020, casares2020, augustino2023, dalzell2023}, conditions that do not carry over to the combinatorial landscape of ILP. Similarly, variational approaches like QAOA or adiabatic annealing offer approximate mappings but have yet to demonstrate systematic performance gains on large, constrained instances without extensive classical feedback loops \cite{booth2021, svensson2023}.

In this work, we depart from hybrid or oracle-based paradigms by proposing a \emph{fully quantum} Metropolis–Hastings algorithm for ILP. 
Every arithmetic subroutine---cost evaluation, constraint counting, and acceptance-probability computation---is implemented as a reversible quantum circuit, avoiding any classical pre- or post-processing. In contrast to previous approaches relying on QRAM or preloaded data, our method performs all calculations \emph{on-the-fly}, enabling coherent updates and transitions directly within the quantum circuit. Furthermore, the algorithm is designed with an all-to-all connectivity topology, allowing direct transitions between arbitrary states. This approach facilitates global exploration of the optimization landscape, thereby significantly mitigating the risk of the search process stagnating in local minima.

The resulting unitary walk explores the feasible polytope coherently, biasing the quantum state amplitude toward high-quality integer solutions. In addition to preserving the massive parallelism of the Metropolis–Hastings framework, this construction opens the door to genuinely quantum enhancements such as amplitude amplification and coherent thermalization, which are unavailable in hybrid or semi-classical samplers. 

A central contribution of this work is the explicit characterization of the logical resource requirements. We derive two key scaling results:
First, the spatial complexity scales as $\mathcal{O}(n \log_2 N)$, where $n$ is the number of variables and $N$ the discretization size. This reflects an efficient logarithmic encoding where auxiliary overheads scale linearly with $n$. Second, the logical depth of the Metropolis step, measured in Toffoli-equivalent gates, scales linearly as $\mathcal{O}(k)$ with the total qubit count $k$. This linear scaling captures the dominant cost of reversible arithmetic and establishes that, while the classical search space grows exponentially, the quantum circuit resources remain polynomial and predictable.

While the method achieves full quantum coherence and predictable scaling, it remains subject to the practical overhead associated with reversible arithmetic and multi-controlled operations, which dominate the Toffoli-equivalent cost. Moreover, the number of annealing steps required to reach thermal equilibrium depends on the specific instance and temperature schedule, introducing a heuristic component typical of annealing-based approaches. 

It is important to note that the efficiency analysis presented in this work refers strictly to the logical circuit complexity, assuming access to a fully fault-tolerant quantum computer. Physical-level overheads related to error correction and qubit connectivity are therefore not included in the reported scaling estimates. Nevertheless, these characteristics are intrinsic to the algorithm’s design and do not compromise its asymptotic efficiency or scalability.

The remainder of this paper is organized as follows. Section~\ref{sec:ILP_fomalization} formalizes the ILP problem and establishes the notation used throughout the paper. Section~\ref{sec:overview} provides a high-level overview of the proposed quantum Metropolis–Hastings algorithm. Section~\ref{sec:restrictions} describes the quantum encoding of linear constraints and their circuit implementation. Section~\ref{sec:quantum_metropolis} details the construction of the full quantum Metropolis–Hastings operator, including proposal, conditional rotation, transition, and reflection subroutines. Finally, Section~\ref{sec:results} presents the numerical simulations that validate the scaling analysis and performance of the method, and Section~\ref{sec:conclusions} summarizes our findings and outlines directions for future research.

\begin{figure*}[t]
  \centering
  \includegraphics[width=\textwidth]{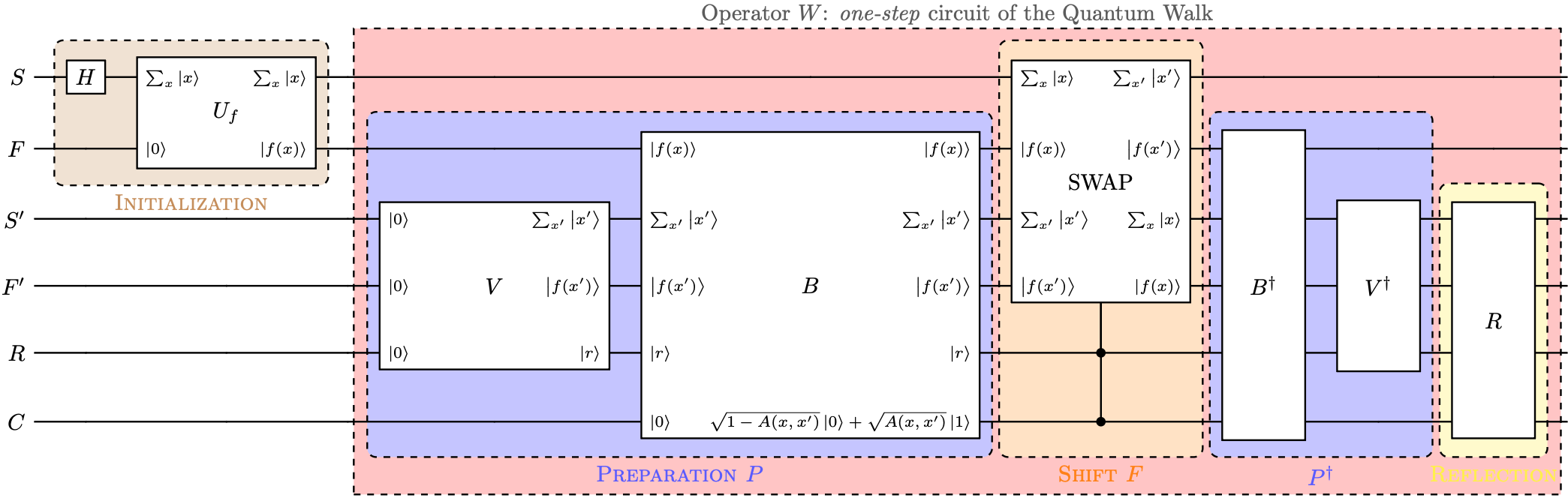}
  \caption{\justifying One iteration of the quantum Metropolis walk. Registers (top to bottom) hold the current position $S$, its cost $F$, a candidate position $S'$, its cost $F'$, the constraint counter $R$ and a coin qubit $C$. The red shaded region encloses the walk operator \(W = R\,P^{\dagger}\,F\,P\).}
  \label{fig:algorithm_outline}
\end{figure*}

\section{Linear Optimization Problems} \label{sec:ILP_fomalization}

Linear programming (LP) is the canonical framework for optimizing a linear objective over a polyhedral feasible region described by linear equality and inequality constraints. 

Integer linear programming strengthens LP by imposing integrality on some or all decision variables, thereby capturing inherently discrete decisions at the cost of a substantial increase in computational difficulty.

\subsection{Canonical definition}
Let $\mathbf{x}=(x_{1},\dots,x_{n})^{\mathsf T}\in\mathbb{R}^{n}$ be the decision vector and
let $\Theta\subseteq\mathbb{R}^{n}$ denote the state space, i.e., the set of all values
$\mathbf{x}$ is allowed to take.

A linear optimization problem in canonical form is

\begin{align}\label{eq:canonicalLP}
  \underset{\mathbf{x}\in\Theta}{\text{minimize}} \quad
    & f(\mathbf{x}) \;=\; c_{0} + c_{1}x_{1}+\dots+c_{n}x_{n}, \\[2pt] \notag
  \text{subject to} \quad
    & g_{i}(\mathbf{x}) \;\ge\; 0,  \quad i = 0,\dots,m-1, \\[2pt]  \notag
    & h_{j}(\mathbf{x}) \;=\; 0,    \quad j = 0,\dots,p-1. 
\end{align}

Throughout this work, we consider $f$, each $g_{i}$, and each $h_{j}$ to be linear functions defined by integer coefficients and constants. Specifically, we assume $c_l \in \mathbb{Z}$ for all $l$, and that the linear constraints are similarly parameterized by integers. This restriction aligns with the fixed-point arithmetic inherent to the quantum adder constructions utilized in our algorithm. The inequalities $g_{i}(\mathbf{x})\ge0$ are called inequality constraints, whereas the equalities $h_{j}(\mathbf{x})=0$ are equality constraints; we take $m\ge0$ and $p\ge0$.

The feasible region is
\begin{equation}
      \Omega \;=\;
  \bigl\{\mathbf{x}\in\Theta : g_{i}(\mathbf{x})\ge0 \;\forall i,\;
  h_{j}(\mathbf{x})=0 \;\forall j\bigr\}\subseteq\Theta .
\end{equation}
By construction, the feasible region $\Omega$ is a convex polytope, obtained as the intersection of half-spaces and hyperplanes defined by the linear constraints.

An optimal solution is a point $\mathbf{x}^\star\in\Omega$ that minimizes (or, in the equivalent maximization form, maximizes) $f(\mathbf{x})$ on $\Omega$.

Because $\max f(\mathbf{x})=-\min[-f(\mathbf{x})]$, every minimization problem can be converted into a maximization problem. Likewise, a “$\le$’’ constraint can be turned into a “$\ge$’’ constraint by simply multiplying by $-1$.

Finally, in integer linear programming, the additional integrality constraint $\mathbf{x} \in \mathbb{Z}^n$ is imposed. This represents the central challenge addressed in this work. In the following sections, we show how such problems can be tackled within a fully quantum framework based on quantum walks.

\section{Algorithm overview} \label{sec:overview}

The quantum Metropolis-Hastings algorithm acts as a unitary embedding of the classical Metropolis–Hastings sampler. While the general framework allows for sampling from arbitrary probability distributions, for optimization purposes we configure the algorithm to produce quantum superpositions whose measurement statistics converge to the Boltzmann distribution $\pi_\beta(\mathbf{x})\propto e^{-\beta f(\mathbf{x})}$. By repeatedly applying a reversible proposal–acceptance block, the algorithm explores the feasible region \(\Omega\) while amplifying the amplitudes of maximal reward solutions. Once the walk has mixed, a measurement yields—with high probability—a state that minimizes the linear objective $f(\mathbf{x})$ in (\ref{eq:canonicalLP}), as can be seen in Figure \ref{fig:3_step_inference}.

In our construction, the reversible proposal–acceptance block is implemented using the framework of coined quantum walks, following the quantum simulated annealing (QSA) algorithm introduced in \cite{PhysRevLett.101.130504}, which ensures convergence to the Boltzmann distribution, but is slightly modified with the randomization technique, resulting in an original scheme to address the ILP problem.

A discrete-time quantum walk is obtained by repeatedly applying a unitary step operator, denoted by $W$, where each application of $W$ corresponds to one step of the walk~\cite{Montanaro:2005rnq}. The operator $W$ must be constructed so as to respect the adjacency structure of the underlying graph.

In our setting, each vertex corresponds to an assignment $\mathbf{x}$ of the $n$ decision variables, encoded in the system register $S$ as
\begin{equation}
    \ket{\mathbf{x}}_{S} \;=\; \ket{x_1}\otimes \ket{x_2}\otimes \cdots \otimes \ket{x_n}.
\end{equation}

To represent transitions between vertices, the system space is extended with an auxiliary \emph{proposal} space. In general graphs the number of neighbors may depend on the vertex, so it is more precise to view this space as having dimension at least the maximum degree of the graph (or, equivalently, to choose a register large enough to index all candidate moves allowed by the proposal rule). We denote this auxiliary register by $S'$, and we interpret it as encoding a proposed neighboring configuration $\ket{\mathbf{x}'}_{S'}$.

In addition, we introduce a \emph{coin} register $C$ consisting of a single qubit, which coherently encodes the transition amplitude associated with the move from the current state $\ket{\mathbf{x}}_{S}$ to the proposed state $\ket{\mathbf{x}'}_{S'}$. Furthermore, two \emph{function} registers $F$ and $F'$ are used to store the objective-function values evaluated at the current and proposed configurations, respectively, i.e., $f(\mathbf{x})$ and $f(\mathbf{x}')$ as defined in Eq.~(\ref{eq:canonicalLP}). Finally, a \emph{constraint-counter} register $R$ records how many of the problem constraints are satisfied by a candidate configuration, and is used to ensure that only feasible transitions contribute to the walk dynamics.

The underlying transition graph is taken to be fully connected: from any configuration $\ket{\mathbf{x}}_{S}$ the proposal register can coherently encode any candidate configuration $\ket{\mathbf{x}'}_{S'}$. Feasibility is not enforced by removing edges from the graph, but by the walk dynamics itself. In particular, the constraint-evaluation circuit and the acceptance filter suppress (or null) amplitudes associated with moves that would violate the problem constraints, so that transitions between feasible and infeasible configurations carry zero effective weight. Thus, while the proposal mechanism remains complete at the level of connectivity, the implemented quantum evolution restricts the walk to the feasible region by construction.

The MH acceptance rule is implemented coherently via the coin register. Given a proposed transition from $\mathbf{x}$ to $\mathbf{x}'$, we define the acceptance probability
\begin{equation}\label{eq:acc_prob_overview}
A(\mathbf{x},\mathbf{x}') \; :=\; \min\!\left\{1,\,\exp\!\big(-\beta\,[f(\mathbf{x}')-f(\mathbf{x})]\big)\right\},
\end{equation}
where $\beta$ represents an annealing schedule. This is the Boltzmann distribution, which we want to sample in order to search for the minimum or maximum values of the objective function. Any function that highlights the maximal values can be used, since the algorithm will attempt to sample it and thus increase the probability of finding these optimal values. In our coined-walk construction, an acceptance subroutine encodes $\sqrt{A(\mathbf{x},\mathbf{x}')}$ as an amplitude on the coin qubit. 

We now outline the one-step quantum walk operator $W$ at a heuristic level, deferring the explicit circuit construction to Section~\ref{sec:quantum_metropolis} and the formal proofs of unitarity and spectral properties to Appendix~\ref{ap:spectral_decomposition}. One step of the quantum Metropolis walk is implemented by the unitary
\begin{equation}\label{eq:W_overview}
W \;=\; R\,P^{\dagger}\,F\,P .
\end{equation}

The structure and ordering of these constituent operators of the walk $W$ mirrors the standard Metropolis--Hastings update rule. First, the \emph{preparation} stage $P$ generates a coherent candidate move in $S'$ and encodes its acceptance amplitude in the coin and auxiliary registers (including feasibility information). Second, the \emph{conditional shift} $F$ applies the proposed update, conditioned jointly on feasibility and on the acceptance register. Third, $P^{\dagger}$ uncomputes the auxiliary information, restoring the ancilla registers to their initial states and ensuring reversibility. Finally, the \emph{reflection} $R$ applies a phase flip about the designated reference subspace, yielding the spectral structure required by the Szegedy-type quantization (in particular, a unique eigenvalue $1$ associated with the target stationary state and a nonzero spectral gap). Appendix~\ref{ap:spectral_decomposition} proves that $W$ is unitary and that its induced measurement statistics reproduce the Metropolis--Hastings acceptance rule. Figure~\ref{fig:algorithm_outline} illustrates one full iteration of the operator $W$.

Operationally:

\begin{enumerate}[label=\arabic*), leftmargin=2.2em]

  \item[\bf (1)] \emph{Initialization.}  
        Prepare a uniform superposition and compute the cost function as the initial input state, 
        \begin{equation}
          \ket{\psi^{(0)}} = \frac{1}{\sqrt{\Theta}} \sum_{\mathbf{x}\in\Theta} 
          \ket{\mathbf{x}}_{S}\ket{f(\mathbf{x})}_{F},
        \end{equation}
        with auxiliary registers set to $\ket{0}$.

  \item[\bf (2)] \emph{Set temperature.}  
        Set the maximum number of annealing steps $Q$ and the maximum number of repetitions per temperature step $T$.
        Fix the inverse temperature schedule $\{\beta_1, \beta_2, \ldots, \beta_Q\}$, typically ranging from $\beta_1 \approx 0$ to a sufficiently large $\beta_Q$.  
        Initialize counter $k \!\leftarrow\! 1$.

  \item[\bf (3)] \emph{Randomized quantum walk.}  
        Draw a random integer 
        \begin{equation}
          t \in \{1, \,T\},
        \end{equation}
        and apply the corresponding block
        \begin{equation}
          \ket{\psi^{(k+1)}} = (W_{k+1})^{t} \ket{\psi^{(k)}},
        \end{equation}
        where $W_{k} = R P_{k}^{\dagger} F P_{k}$.  
        Here $P_k$ encodes the thermalization step defined by $\beta_k$.  
        The operation $(W_k)^t$ performs a randomized quantum walk whose length is controlled by the random variable $t$, enhancing mixing within the feasible subspace.

  \item[\bf (4)] \emph{Partial measurement and update.}  
        Perform a partial measurement on the coin qubit $\mathcal{H}_C$ and the movement register $\mathcal{H}_{S'}$ and discard the result.
        Increment the temperature step counter $k \!\leftarrow\! k + 1$.  
        If $k < Q$, return to Step~(3) and continue the annealing process with the next $\beta_k$.

  \item[\bf (5)] \emph{Termination.}  
        When $k = Q$, the process halts and the final state $\ket{\psi^{(Q)}}$ approximates the Boltzmann distribution
        \begin{equation}
          \pi_{\beta_Q}(\mathbf{x}) \;\propto\; e^{-\beta_Q f(\mathbf{x})},
        \end{equation}
        restricted to the feasible region $\Omega$.
\end{enumerate}

Following Ref.~\cite{PhysRevLett.101.130504}, we define an annealing schedule $\{\beta_k\}_{k=1}^{Q}$ that dictates the rate of cooling. For each $\beta_k$, we construct a quantum walk operator $W_k$ corresponding to the classical transition matrix $M_{\beta_k}$, whose stationary state is the quantum Gibbs distribution
\begin{equation}
\pi_{\beta_k}(\mathbf{x}) \propto e^{-\beta_k f(\mathbf{x})}.
\end{equation}
The algorithm sequentially prepares these thermal states through the iterative transformation $\ket{\psi^{(k)}} \longrightarrow \ket{\psi^{(k+1)}}$. Crucially, the step sizes $\Delta\beta_k = \beta_{k+1} - \beta_k$ are chosen to be sufficiently small to ensure a high overlap between consecutive thermal states,
\begin{equation}
\big|\!\braket{\psi^{(k)}|\psi^{(k+1)}}\!\big|^2 \geq 1 - \epsilon,
\end{equation}
for some small $\epsilon > 0$. This condition ensures that the sequence of projections and evolutions effectively drags the system toward the low-cost states of $f(\mathbf{x})$ with high probability.

At high temperatures (small $\beta$), nearly all transitions are accepted except those violating feasibility constraints, allowing a broad exploration of the search space. As the temperature decreases, the acceptance probability becomes increasingly selective, and in the low-temperature limit, only transitions that reduce the objective function value are accepted. This gradual cooling emulates the classical simulated-annealing process, but here it occurs coherently within a reversible quantum framework.

Unlike the classical Metropolis–Hastings procedure, where convergence relies on the empirical estimation of a mixing time $T_{\mathrm{mix}}$ \cite{GIORDANO2026170305}, the quantum approach relies on spectral filtering via quantum phase estimation. The walk operator $W_k$ possesses a unique eigenstate with eigenvalue $1$ (corresponding to eigenphase $\theta=0$) while all other eigenvalues are separated by a finite spectral gap $\Delta$ (see Appendix \ref{ap:spectral_decomposition}). Consequently, the computational cost is determined by the requirement to resolve this gap. Specifically, distinguishing the target eigenstate requires a phase estimation precision of $2^{-q} \lesssim \Delta$, which implies $\mathcal{O}(1/\Delta)$ applications of $W_k$ \cite{Kitaev1995QuantumMA, 10.1098/rspa.1998.0164}. Provided this precision is met, the QPE subroutine projects the state $\ket{\psi^{(k)}}$ onto the desired Gibbs distribution with high probability, effectively filtering out the excited states that correspond to the classical mixing transient.

After $Q$ annealing steps, with sufficient spectral resolution, the final state $\ket{\psi^{(Q)}}$ approximates the quantum Gibbs state associated with the Boltzmann weights $\pi_{\beta_Q}(\mathbf{x}) \propto e^{-\beta_Q f(\mathbf{x})}$. A projective measurement on the system register $S$ then samples configurations $\mathbf{x}$ according to this low-temperature thermal distribution, yielding solutions whose energies lie within $\epsilon$ of the global optimum.

In the original QSA formulation~\cite{PhysRevLett.101.130504}, the total number of required quantum walk steps to accurately sample the thermal distribution scales as
\begin{equation}
    \mathcal{N}_{\mathrm{QSA}} = 
    \mathcal{O}\!\left[
    \left( \frac{E_M}{\gamma} \right)^{2}
    \frac{\log^{2}(d/\epsilon)\,\log d}{\epsilon \sqrt{\delta}}
    \right],
\end{equation}
where $E_M$ is the bound on the energy function range, $d$ represents the dimension of the state space, and $\epsilon$ is the desired precision. Crucially, $\delta$ denotes the spectral gap of the \emph{classical} transition matrix. Due to the quadratic speedup property of Szegedy-type quantum walks \cite{1366222, PhysRevResearch.5.033059}, the spectral gap of the unitary operator $W_k$—denoted previously as $\Delta$—scales as $\Delta \sim \sqrt{\delta}$. This relationship dictates the phase estimation precision required to distinguish the ground state, leading to the $1/\sqrt{\delta}$ factor in the complexity, in contrast to the $1/\delta$ scaling of classical mixing times.

However, in the \emph{optimization} setting considered here, the requirements are relaxed compared to exact sampling. Since our goal is to amplify the probability of the optimal configuration rather than to strictly sample the full Gibbs distribution, the number of annealing steps need not saturate the theoretical bound derived for sampling. In practice, the process acts as a heuristic filter: the sequential application of $W_k$ operators, combined with the partial measurement scheme \cite{PhysRevResearch.6.043014}, progressively concentrates probability amplitude into the lowest-energy subspace. Unlike Grover-type amplitude amplification, where precise timing is critical to avoid over-rotation \cite{PhysRevLett.79.325, PhysRevA.62.062303}, this annealing approach is robust: extending the schedule typically improves—or at least preserves—the success probability, rather than oscillating away from the solution.

Therefore, while the theoretical upper bound in~\cite{PhysRevLett.101.130504} remains the rigorous limit for exact Gibbs sampling, in the optimization-oriented version of QSA implemented here, convergence to near-optimal solutions is observed empirically to scale more favorably. We find that a polynomial or polylogarithmic number of temperature steps often suffices to reach the feasible ground state. In this sense, the annealing schedule serves as a tunable heuristic: slower cooling (smaller $\Delta\beta$) reliably enhances the fidelity of the final state with respect to the global minimum.

\begin{figure}[t]
    \includegraphics[width=0.25\textwidth]{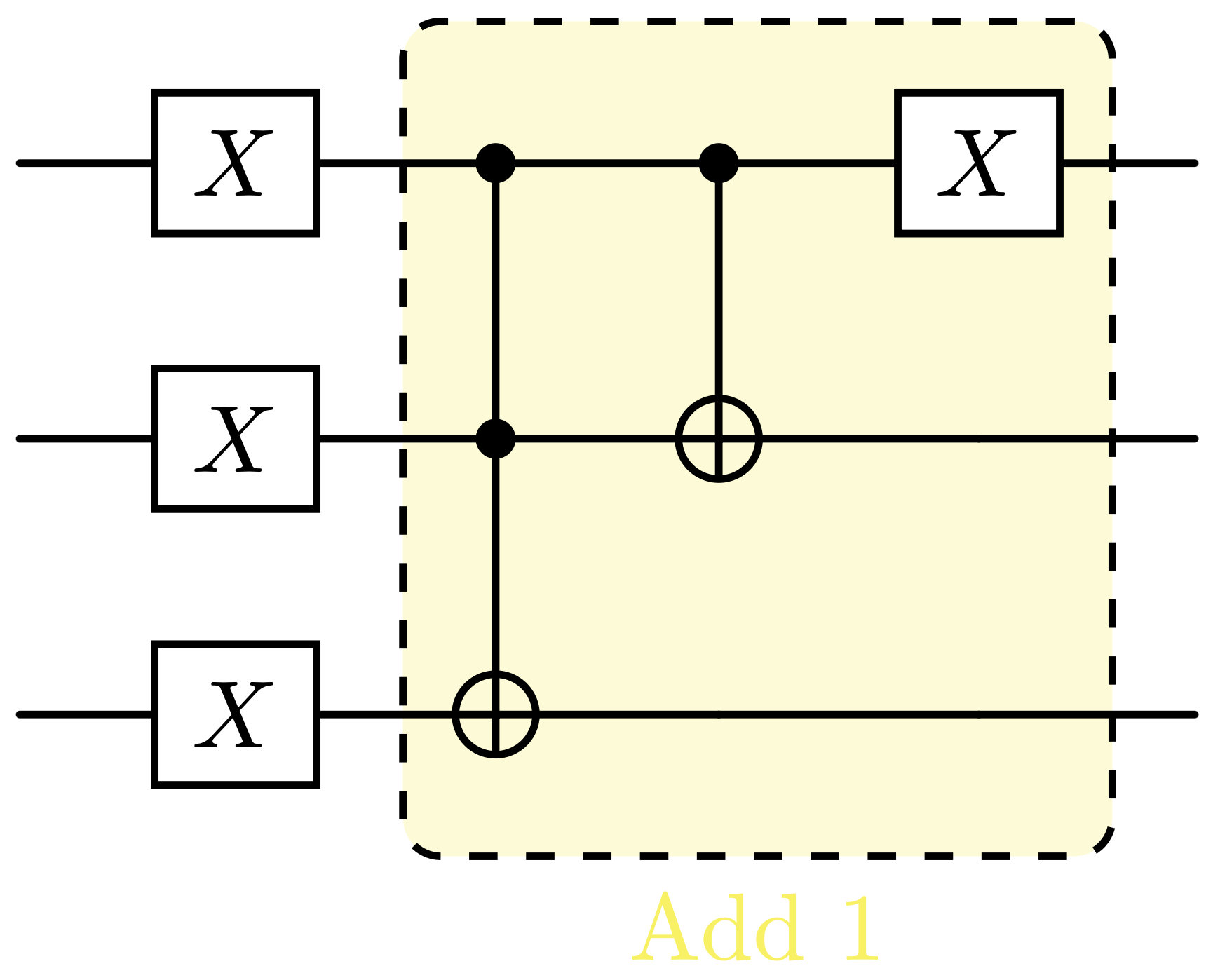}
    \captionsetup{justification=justified,singlelinecheck=false,font=small}
    \caption{\justifying Quantum circuit implementing two's-complement negation of an $n$-qubit register ($n=3$ shown). The circuit comprises two stages: (i) bit-wise NOT operations (left-hand $X$ gates) that flip all qubits, and (ii) a cascade of multi-controlled NOT gates (shaded region) that increments the result by one, together realizing the two's-complement negation $- x = \overline{x} + 1$.}
    \label{fig:twos_complement}
\end{figure}

\section{Quantization of the Restrictions} \label{sec:restrictions}

 \begin{figure*}[t]
  \centering
  \includegraphics[width=1\linewidth]{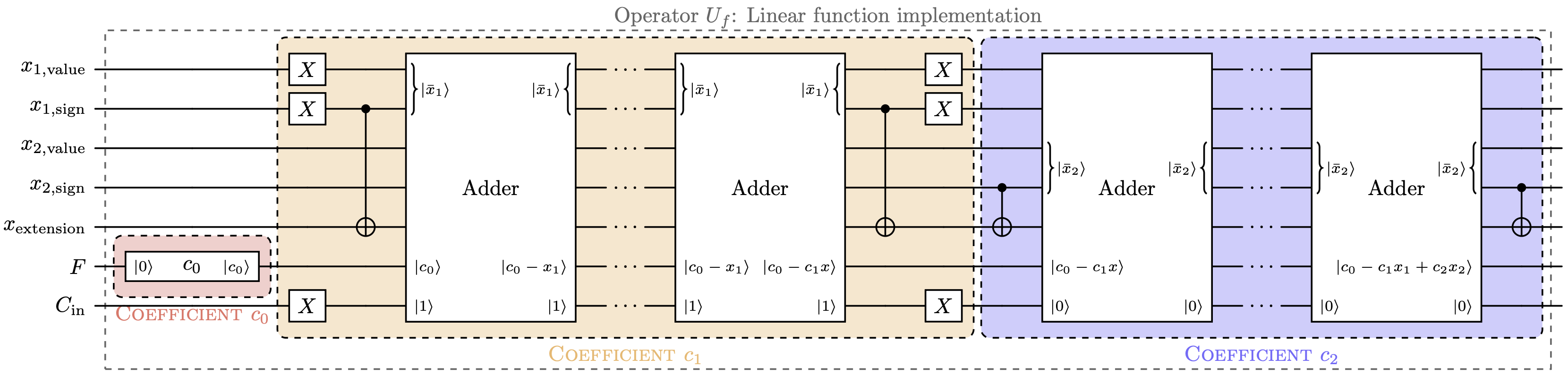}
  \captionsetup{justification=justified, font=small}
  \caption{\justifying Synthesis of the unitary $U_f$ for the two-variable linear form $f(x_{1},x_{2}) = c_{0}-c_{1}x_{1}+c_{2}x_{2}$ with $c_{i}>0$. Register $F$ is first loaded with the constant term $c_{0}$ (red box). \textbf{Orange block:} subtraction of $c_{1}x_{1}$. The value qubits of $x_{1}$ and the global carry-in are flipped with $X$ gates, the sign bit is fan-out via CNOTs to sign-extend $x_{1}$ up to the number of qubits of register $F$, and the CDKM full adder is applied $c_{1}$ times, yielding $c_{0}-c_{1}x_{1}$. \textbf{Blue block:} addition of $c_{2}x_{2}$. Because $c_{2}$ is positive, only sign extension is required before invoking the adder $c_{2}$ times. Upon completion the accumulator holds $c_{0}-c_{1}x_{1}+c_{2}x_{2}$ and all ancillas (carry and extension lines) are returned to $\ket{0}$, ensuring that $U_f$ is reversible.}
  \label{fig:quantum_adder}
\end{figure*}

By definition, linear programming restricts both the objective function and every constraint to linear forms; evaluating any of them on a quantum register therefore boils down to computing weighted sums of integer variables. Addition is thus the core arithmetic primitive, and the literature offers several depth and qubit efficient reversible adders to implement it \cite{PhysRevA.54.147, wang2025quantum}. Building on these primitives, the remainder of this section assembles the higher‑level unitary operators required by our algorithm.

\subsection{Encoding}

As we are going to work with integers, we have to represent both positive and negative numbers. In order to represent negative numbers, we will use the binary encoding two's complement. In this way, the subtraction can be performed according to:
\begin{equation}
    a - b = a + \bar{b} + 1,
\end{equation}
where $\bar{x}$ is the negation of $x$ bit by bit, and is nothing more than the consecutive application of $X$ gates. Increasing the unit can be done simply with multi-controlled gates. This simple circuit is shown in Figure \ref{fig:twos_complement}.

\subsection{Quantum Circuits for Linear Functions}
Multiplication by a classical integer coefficient is implemented as a sequence of conditional additions; hence the only primitive we need is a reversible adder.  We adopt the Cuccaro–Draper–Kutin–Moulton ripple‑carry design (CDKM) \cite{Cuccaro:2004xxx} because of its qubit‑optimal footprint; any other reversible adder can be substituted without altering the logic that follows.

The CDKM adder natively performs
\begin{equation}
    \ket{a}\ket{b}\!\mapsto\!\ket{a}\ket{a+b\bmod 2^{m}},
\end{equation}
on two equally sized $m$ qubit registers. To incorporate subtraction with the CDKM adder we negate the second operand in two steps. First apply \(X\) gates to all \(m\) qubits of \(b\) (bit‑wise NOT). Then, invoke the fixed‑adder version of CDKM with the carry‑in initialized to \(\ket{1}\).  This adds the required +1 and yields \(a + \overline{b} + 1 = a-b\).

\begin{figure*}[t]
    \centering \includegraphics[width=1\textwidth]{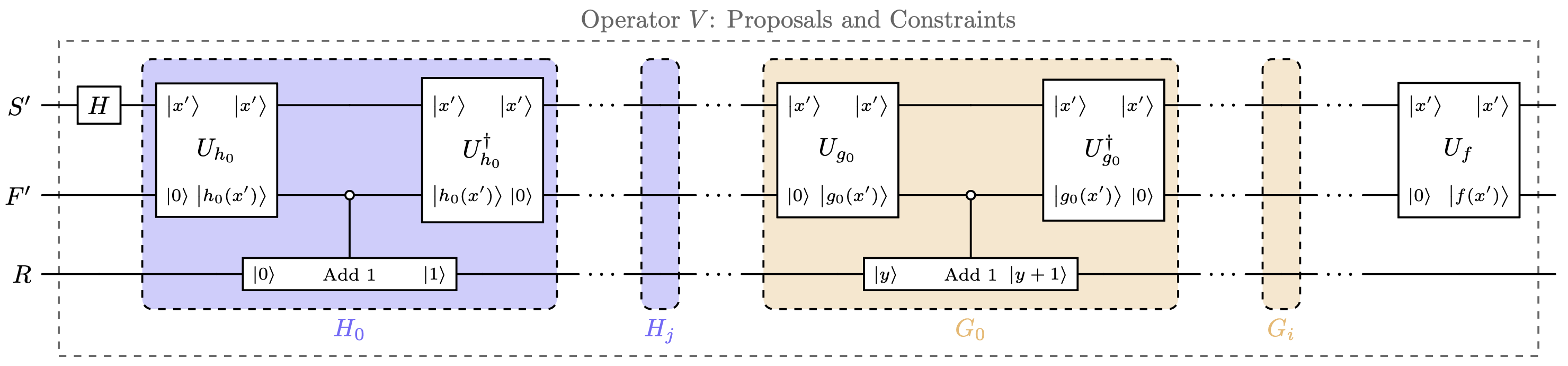}
    \captionsetup{justification=justified, font=small}
    \caption{\justifying Explicit circuit for the operator $V$ of the one-step operator $W$. For every candidate vector $\mathbf{x'}$ in $S'$ the unitary first computes the constraint values in the workspace register $F'$ and then updates the counter register $R$. Controls on the equality blocks $H_j$ act on \emph{all} qubits of $F'$, whereas the inequality blocks $G_i$ need only the sign qubit.  The final state of $R$ encodes the number of satisfied constraints, being $m{+}p$ for feasible points.}
    \label{fig:V_operator}
\end{figure*}

The NOT+carry‑in trick avoids an extra adder since the CDKM fixed adder always works with one carry in qubit.

All CDKM inputs must have the same width.  If variable \(x_i\) is stored in \(k\) qubits but the accumulator uses \(m>k\) qubits, we \emph{sign‑extend} \(x_i\) by copying its sign bit into the \(m-k\) highest positions; this is done with CNOTs controlled by the sign qubit. All ancillas are uncomputed locally, so the transformation remains unitary.

Let us introduce the reversible operations associated to the integer linear optimization problem by means of the following unitary operators:
\begin{align} \label{function_operators}
    \begin{split}
        U_{g_i} \ket{\mathbf{x}}_{S'}\ket{0}_{F'} &:=  \ket{\mathbf{x}}_{S'}\ket{g_i (\mathbf{x})}_{F'},\\
        U_{h_j} \ket{\mathbf{x}}_{S'}\ket{0}_{F'} &:=  \ket{\mathbf{x}}_{S'}\ket{h_j (\mathbf{x})}_{F'},\\
        U_f \ket{\mathbf{x}}_{S'}\ket{0}_{F'} &:=  \ket{\mathbf{x}}_{S'}\ket{f (\mathbf{x})}_{F'}.\\
    \end{split}
\end{align}
These operators compute the functions (\ref{eq:canonicalLP}) in quantum states. 

Let us see explicitly the construction for $f(x)$ in (\ref{eq:canonicalLP}), being a general scheme for any linear function. This process is shown in Figure \ref{fig:quantum_adder}. 

First we initialize the circuit with the independent term $c_0$ with gates $X$ so we get the state $\ket{\mathbf{x}}_{S'}\ket{c_0}_{F'}$. The next step is to add the variable $x_1$. To do this, we need to take into account the sign of the coefficient. If the coefficient is positive, we perform the sign extension and apply $c_1$ times the CDKM adder. On the other hand, if the coefficient is negative, we negate the $x_1$ register with $X$ gates, set the carry in to $\ket{1}$ and perform the sign extension. We then apply $|c_1|$ times the adder and descompute the sign extension and the $X$ gates. In this way we reach the state:
\begin{equation}
    \ket{\mathbf{x}}_{S'}\ket{c_0}_{F'} \mapsto \ket{\mathbf{x}}_{S'}\ket{c_0 + c_1x_1}_{F'}.
\end{equation}
Repeating this process for each of the $n$ variables $x_i$, each with its own coefficient $c_i$, we finally arrive at the state $\ket{\mathbf{x}}_{S'}\ket{f (\mathbf{x})}_{F'}$.

Each adder gate uses \(\mathcal{O}(m)\) Toffoli gates for an \(m\)-bit adder \cite{Cuccaro:2004xxx}, which is applied to \(n\) variables $x_i$, each with its coefficient $c_i$, giving a total of \(\mathcal{O}(m\sum_i |c_i| )\) Toffoli gates.

Although these arithmetic elements do not accelerate individual operations versus classical logic, the advantage will lie in their coherent evaluation over a superposition of all \(\mathbf{x}\in\Theta\), essential for the quantum Metropolis walk developed in Sec.~\ref{sec:quantum_metropolis}.

\subsection{Quantum Circuits for Constraint Implementation}

We now design a reversible routine that, for every candidate point $\mathbf{x}$ held in the candidate state register $S'$, evaluates all $p$ equality constraints $h_j$ and all $m$ inequality constraints $g_i$ and stores the result in a dedicated counter register $R$. The role of $R$ is purely to count how many of the $(m+p)$ constraints are satisfied by~$\mathbf{x}$; it needs $\lceil\log_{2}(m+p)\rceil$ qubits. The purpose of this routine is to perform the following operation:

\begin{equation}\label{C_operator}
\ket{\mathbf{x}}_{S'}\ket{0}_{R}\;\longmapsto\;
  \begin{cases}
     \ket{\mathbf{x}}_{S'}\,\ket{m+p}_{R}, & \mathbf{x}\in\Omega,\\[4pt]
     \ket{\mathbf{x}}_{S'}\,\ket{r}_{R},   & \mathbf{x}\notin\Omega,\\
  \end{cases}
\end{equation}
where $\Omega$ is the feasible region and $r\le m+p$ is the number of constraints met by the infeasible point $\mathbf{x}$.

The registers involved are $S'$ – candidate state register, $\ket{\mathbf{x}}_{S'}$; $F'$ – workspace register used to compute each $g_i(\mathbf{x})$ or $h_j(\mathbf{x})$ via the unitary operators $U_{g_i},U_{h_j}$ and the cost function of Eq.~\eqref{function_operators}; $R$  – counter register described above.

This explicit construction will be explained in more detail in the following section in order to arrive at (\ref{C_operator}). All the following developments can be found in Figure \ref{fig:V_operator}.

\subsection{Equality constraints}

To study the equality constraints, the first step is to apply the first constraint, as follows:
\begin{equation}
    U_{h_0} \ket{\mathbf{x}}_{S'} \ket{0}_{F'} \ket{0}_{R} =  \ket{\mathbf{x}}_{S'} \ket{h_0 (\mathbf{x})}_{F'} \ket{0}_{R}.
\end{equation}
As the condition $h_0$ is an equality condition, it will be fulfilled if the $F'$ register is $\ket{0}_{F'}$. Consequently, we increment the $R$ register by 1 (resulting in $\ket{0\dots01}_{R}$) if the constraint is met; otherwise, the register remains in $\ket{0}_{R}$. For increasing the unit the \texttt{Add 1} circuit of the Figure \ref{fig:twos_complement} can be used. For this we define the operator $A_{\text{eq}}$, which acts as follows:

\begin{align}
\small
    A_{\text{eq}} & \ket{\mathbf{x}}_{S'} \ket{h_0 (\mathbf{x})}_{F'} \ket{0}_{R} := \notag\\
    &=  \begin{cases}
         \ket{\mathbf{x}}_{S'} \ket{h_0 (\mathbf{x})}_{F'} \ket{0...01}_{R} & \text{if } \ket{h_0 (\mathbf{x})}_{F'} = \ket{0}_{F'}, \\
         \ket{\mathbf{x}}_{S'} \ket{h_0 (\mathbf{x})}_{F'} \ket{0...00}_{R} & \text{if } \ket{h_0 (\mathbf{x})}_{F'} \neq \ket{0}_{F'}.
       \end{cases}
\end{align}
And this comparison between $\ket{h_0 (\mathbf{x})}_{F'}$ and $\ket{0}_{F'}$ is as simple as applying X on every qubit of $F'$, then applying Multi Controlled \texttt{Add 1} gate of Figure \ref{fig:twos_complement} controlled by all the qubits of the $F'$, and finally uncompute the transformation applying $X$ on every qubit of $F'$. This controlled operation is sometimes referred to as Anti-controlled gate, being represented as an empty circle in the Figure \ref{fig:V_operator}.

Before applying the next constraint, we must uncompute $U_{h_0}$ to reset the $F'$ register to $\ket{0}_{F'}$. This step enables us to store the value of the next function in the same register, optimizing qubit usage. This is simply applying $U_{h_0}^{\dagger}$:
\begin{equation}
    U_{h_0}^{\dagger} \ket{\mathbf{x}}_{S'} \ket{h_0 (\mathbf{x})}_{F'} \ket{r}_{R}  =  \ket{\mathbf{x}}_{S'} \ket{0}_{F'} \ket{r}_{R}
\end{equation}

For the next constraint, we apply the exact same procedure as described above. Specifically:
\begin{equation}
    U_{h_1} \ket{\mathbf{x}}_{S'} \ket{0}_{F'} \ket{r}_{R} =  \ket{\mathbf{x}}_{S'} \ket{h_1 (\mathbf{x})}_{F'} \ket{r}_{R}.
\end{equation}

Following the same logic, the constraint is satisfied if the $F'$ register remains in the state $\ket{0}_{F'}$:
\begin{align}
\small
    A_{\text{eq}} & \ket{\mathbf{x}}_{S'} \ket{h_1 (\mathbf{x})}_{F'} \ket{r}_{R} = \notag\\
    &=  \begin{cases}
         \ket{\mathbf{x}}_{S'} \ket{h_1 (\mathbf{x})}_{F'} \ket{r+1}_{R} & \text{if } \ket{h_1 (\mathbf{x})}_{F'} = \ket{0}_{F'}, \\
         \ket{\mathbf{x}}_{S'} \ket{h_1 (\mathbf{x})}_{F'} \ket{r}_{R} & \text{if } \ket{h_1 (\mathbf{x})}_{F'} \neq \ket{0}_{F'}.
       \end{cases}
\end{align}

Finally, we uncompute by applying $U_{h_1}^{\dagger}$ to reset the $F'$ register to $\ket{0}_{F'}$, thereby allowing for the evaluation of subsequent constraints:

\begin{equation}
    U_{h_1}^{\dagger} \ket{\mathbf{x}}_{S'} \ket{h_1 (\mathbf{x})}_{F'} \ket{r}_{R}  =  \ket{\mathbf{x}}_{S'} \ket{0}_{F'} \ket{r}_{R}.
\end{equation}

\begin{figure*}[t]
    \centering
    \includegraphics[width=0.85\textwidth]{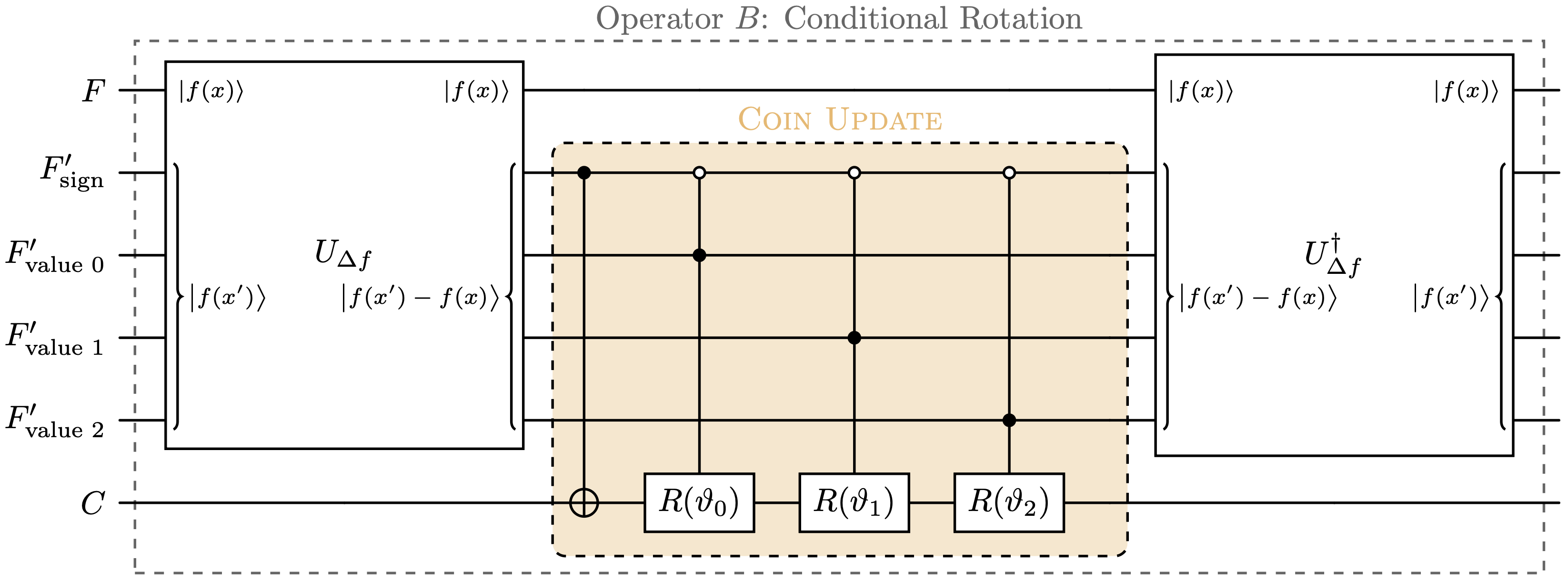}
    \captionsetup{justification=justified, font=small}
    \caption{\justifying Explicit circuit construction of the operator $B$ in Eq.~(\ref{eq:rotation}). The Metropolis acceptance amplitude is encoded into the coin register by applying controlled $R_y$ rotations conditioned qubit-by-qubit on the discretized energy-difference register $\Delta f$ stored in $F'$. For each annealing step $k$, the rotation angles $\{\theta^{(k)}_j\}$ are determined by the inverse temperature $\beta_k$ through the prescribed approximation scheme. This qubit-wise construction avoids implementing an explicit exponential subroutine while enabling a $\beta$-dependent thermalization of the acceptance rule.}
    \label{fig:B_operator}
\end{figure*}

We define the composite operator for the $j$-th constraint, denoted as $H_j$, by:
\begin{equation}
    H_j := U_{h_j}^{\dagger} A_{\text{eq}} U_{h_j}.
\end{equation}

We provide explicit constructions for each $H_j$, which implement the following transformation:
\begin{equation} \label{H_operator}
    H_j \ket{\mathbf{x}}_{S'}\ket{0}_{F'}  \ket{r}_{R} =  \begin{cases}
         \ket{\mathbf{x}}_{S'}\ket{0}_{F'}  \ket{r+1}_{R} & \text{if } \mathbf{x} \in \Omega_j, \\
         \ket{\mathbf{x}}_{S'}\ket{0}_{F'}  \ket{r}_{R} & \text{if } \mathbf{x} \notin \Omega_j.
       \end{cases}
\end{equation}
Here, $\Omega_j$ denotes the feasible region defined by the $j$-th equality constraint. This approach yields a clear modular scheme, the architecture of which is illustrated in Figure \ref{fig:V_operator}.

\subsection{Inequality constraints}

Having established the scheme for equality constraints, we now address the construction of the operators for inequality constraints, denoted by $g_{i}$. We define an operator $G_i$, analogous to $H_j$ in Eq.~(\ref{H_operator}), such that:

\begin{equation} \label{G_operator}
    G_i \ket{\mathbf{x}}_{S'}\ket{0}_{F'}  \ket{r}_{R} :=  \begin{cases}
         \ket{\mathbf{x}}_{S'}\ket{0}_{F'}  \ket{r+1}_{R} & \text{if } \mathbf{x} \in \Omega_i, \\
         \ket{\mathbf{x}}_{S'}\ket{0}_{F'}  \ket{r}_{R} & \text{if } \mathbf{x} \notin \Omega_i,
       \end{cases}
\end{equation}
where $\Omega_i$ corresponds to the feasible region satisfied by the $i$-th inequality constraint. The construction of this operator is straightforward. First, we compute the constraint function:
\begin{equation}
    U_{g_0} \ket{\mathbf{x}}_{S'} \ket{0}_{F'} \ket{r}_{R} =  \ket{\mathbf{x}}_{S'}\ket{g_0 (\mathbf{x})}_{F'} \ket{r}_{R}.
\end{equation}

Since $g_i$ represents a greater-than or equal to constraint ($\ge 0$), it is satisfied if the value in the $F'$ register is non-negative. In the two's complement representation, the most significant bit (MSB) acts as the sign bit: a state represents a non-negative value if and only if the MSB is $\ket{0}$. Consequently, the condition corresponds to the $F'$ register being in a state of the form $\ket{0 \dots}_{F'}$.

Thus, we increment the ancilla register $R$ (mapping $\ket{r}_R \to \ket{r+1}_R$) if the constraint is met, leaving it unchanged otherwise. This conditional increment is performed using the same \texttt{Add 1} circuit shown in Figure \ref{fig:twos_complement}. We define the conditional operator $A_{\text{ineq}}$ as follows:

\begin{align}
\small
    & A_{\text{ineq}}  \ket{\mathbf{x}}_{S'} \ket{g_0 (\mathbf{x})}_{F'} \ket{r}_{R} := \notag\\
    &= \begin{cases}
         \ket{\mathbf{x}}_{S'} \ket{g_0 (\mathbf{x})}_{F'} \ket{r+1}_{R} & \text{if } \ket{g_0 (\mathbf{x})}_{F'} = \ket{0 ...}_{F'}, \\
         \ket{\mathbf{x}}_{S'} \ket{g_0 (\mathbf{x})}_{F'} \ket{r}_{R} & \text{if } \ket{g_0 (\mathbf{x})}_{F'} \neq \ket{0 ...}_{F'}.
       \end{cases}
\end{align}

Verifying this condition is computationally less expensive than checking for equality, as it requires querying only the first qubit (MSB) of the $F'$ register. Implementation-wise, since the standard control activates on $\ket{1}$ but we require activation on $\ket{0}$ (positive sign), we apply an $X$ gate to the MSB, perform the controlled \texttt{Add 1}, and then apply $X$ again to restore the state. This is equivalent to an open-controlled (zero-controlled) operation.

Finally, we uncompute $U_{g_i}$ to reset the $F'$ register. The complete construction for any inequality constraint $G_i$ is given by:

\begin{equation}
    G_i := U_{g_i}^{\dagger} A_{\text{ineq}} U_{g_i}.
\end{equation}

The full constraint-checking subroutine, corresponding to Eq.~(\ref{C_operator}), is composed as follows:

\begin{equation}\label{eq:restrictions}
\small
         G_{m-1} \cdots G_1 G_0 H_{p-1} \cdots H_1 H_0 =  \prod_{i=0}^{m - 1} G_{m - 1 -i} \prod_{j=0}^{p - 1} H_{p - 1 -j},
\end{equation}
ensuring that the counter ancilla reaches the state $\ket{m+p}_{R}$ if and only if all constraints are fulfilled.

From an implementation standpoint, it is worth emphasizing that the equality–checking subroutine $H_j$ generally incurs a higher logical cost than its inequality counterpart $G_i$. In particular, testing equality requires a full multi-qubit comparison with $\ket{0}_{F'}$, which translates into multi-controlled increments on the counter register $R$ and therefore leads to a quadratic Toffoli overhead in the bit width of $F'$. By contrast, inequality constraints of the form $g_i(\mathbf{x}) \ge 0$ reduce to a sign test in two’s-complement representation and can be implemented by querying only the most significant bit, yielding a linear logical cost.

This asymmetry suggests a more resource-efficient formulation in which equality constraints are recast as pairs of inequalities.  The following statement formalizes this observation.

\begin{theorem}[Resource-efficient equality elimination]
Consider an ILP instance in canonical form~(\ref{eq:canonicalLP}) with $m \ge 0$ inequality constraints $g_i(\mathbf{x}) \ge 0$ and $p \ge 0$ equality constraints $h_j(\mathbf{x}) = 0$, and let $\Omega$ denote its feasible region. Under the constraint-encoding scheme of Section~\ref{sec:restrictions}, there exists an equivalent formulation of the same problem with
\begin{equation}
    m' \;=\; m + 2p,
\end{equation}
pure inequality constraints only, obtained by replacing each equality constraint $h_j(\mathbf{x}) = 0$ with the pair
\begin{equation}
    h_j(\mathbf{x}) \ge 0,
    \qquad
    -h_j(\mathbf{x}) \ge 0.
\end{equation}
This reformulation preserves the feasible region $\Omega$ exactly and admits a global constraint-checking operator of the form Eq.~(\ref{eq:restrictions}) that can be synthesized exclusively from $G_i$-type inequality subroutines, thereby yielding a strictly lower asymptotic Toffoli overhead than any implementation that employs both $G_i$- and $H_j$-type blocks.
\end{theorem}

\begin{proof}
Let $h:\Theta \rightarrow \mathbb{Z}$ be an integer-valued constraint function. For any $\mathbf{x} \in \Theta$,
\begin{equation}
    h(\mathbf{x}) = 0 
    \;\Longrightarrow\;
    h(\mathbf{x}) \ge 0 \;\wedge\; h(\mathbf{x}) \le 0,
\end{equation}
which is immediate. Conversely, if $h(\mathbf{x}) \ge 0$ and $h(\mathbf{x}) \le 0$ hold simultaneously, the only integer satisfying both inequalities is $h(\mathbf{x}) = 0$. Hence,
\begin{equation}
    h(\mathbf{x}) = 0 
    \;\Longleftrightarrow\;
    \bigl(h(\mathbf{x}) \ge 0\bigr) \;\wedge\;
    \bigl(h(\mathbf{x}) \le 0\bigr),
\end{equation}
and, using $h(\mathbf{x}) \le 0 \iff -h(\mathbf{x}) \ge 0$,
\begin{equation}
    h(\mathbf{x}) = 0 
    \;\Longleftrightarrow\;
    \bigl(h(\mathbf{x}) \ge 0\bigr) \;\wedge\;
    \bigl(-h(\mathbf{x}) \ge 0\bigr).
\end{equation}

\begin{table*}[t]
\centering
\renewcommand{\arraystretch}{1.2}
\setlength{\tabcolsep}{8pt}
\small
\begin{tabular}{|c||l|c|}
\hline
\textbf{Register} & \textbf{Description} & \textbf{Qubits} \\
\hline \hline
$S$   & System register encoding $\Theta$ & $n\log_2N$ \\[2pt]
$S'$  & Possible movements & $n\log_2N$ \\[2pt]
$F$   & Function evaluation of $\ket{\textbf{x}}_S$& $\mathcal{O}(\log_2 (N\sum_k |c_k|) )$ \\[2pt]
$F'$  & Function evaluation of $\ket{\textbf{x}'}_{S'}$ & $\mathcal{O}( \log_2 (N \sum_k |c_k) )$ \\[2pt]
$R$   & Constraints counter & $\ceil{\log_2{(m')}}$ \\[2pt]
$C$   & Coin & $1$ \\[2pt]
\hline
\end{tabular}
\caption{\justifying Summary of the quantum registers used in the algorithm, including their logical meaning and the corresponding number of qubits, except for a constant number of ancillas needed for the adders. Each variable $x_i$ searches the interval $x_i \in [-N, N -1]$ for $i = 1, \ldots, n$. There are $m$ inequailty and $p$ equality constraints, with $m'=m+2p$.}
\label{tab:registers}
\end{table*}

Applying this equivalence to each equality constraint $h_j(\mathbf{x}) = 0$ shows that replacing it by the pair $h_j(\mathbf{x}) \ge 0$ and $-h_j(\mathbf{x}) \ge 0$ does not alter the set of feasible points, and therefore the feasible region $\Omega$ is preserved exactly. The reformulated instance thus contains $m' = m + 2p$ pure inequality constraints.

From an implementation perspective, Section~\ref{sec:restrictions} shows that the equality-checking subroutine $H_j$ requires a full multi-qubit comparison of the workspace register $F'$ with $\ket{0}_{F'}$, realized through multi-controlled increments on the counter register $R$. This leads to a quadratic Toffoli overhead in the bit width of $F'$. By contrast, each inequality constraint $g_i(\mathbf{x}) \ge 0$ is implemented by a $G_i$-type block that reduces to a sign test in two’s-complement representation, querying only the most significant bit of $F'$, and therefore incurs a linear Toffoli cost in the same bit width.

Replacing every equality block $H_j$ with two inequality blocks of type $G_i$ thus transforms each quadratic-cost subroutine into two linear-cost 
ones. At the level of the global constraint-checking operator Eq.~(\ref{eq:restrictions}), this modification strictly reduces the 
leading asymptotic Toffoli overhead while leaving all other circuit components unchanged. Consequently, the overall constraint-checking 
stage exhibits a strictly lower asymptotic cost, while the feasible region $\Omega$ remains identical.
\end{proof}

\section{Quantum Metropolis Approach} \label{sec:quantum_metropolis}

With the constraint-evaluation subroutines in place (Section~\ref{sec:restrictions}), we now move from the high-level description of the algorithm to its explicit circuit-level realization. As introduced in Section~\ref{sec:overview}, the one-step update operator is the unitary $W$ in Eq.~(\ref{eq:W_overview}), acting on the joint register space $\{S,S',F,F',R,C\}$ (see Section~\ref{sec:overview} and Table~\ref{tab:registers}), and repeated applications of this step induce a discrete-time quantum walk whose measurement statistics reproduce Metropolis--Hastings acceptance filtering.

In this section, we present the concrete construction of the walk, refining the architecture proposed in~\cite{Campos2023, Escrig2023, Escrig2025} to support explicit constraint handling. We specify each constituent subroutine---proposal generation, acceptance-amplitude encoding, conditional state transition, and reflection---and discuss the logical resources required to implement them as reversible quantum circuits.
Figure~\ref{fig:algorithm_outline} provides a schematic overview of one full application of the step operator and its register flow.

Operationally:

\begin{enumerate}[label=\textbf{\alph*}), leftmargin=1.8em, itemsep=2pt]
\item \textbf{Preparation (\(P=B\,V\)).}
      \begin{itemize}[leftmargin=1.6em, itemsep=1pt]
        \item[·] \emph{Proposal (\(V\)).}  Generates a superposition of candidate states 
              \(\mathbf{x}'\in\Theta\) in the move register \(S'\) and evaluates
              \(\bigl(f(\mathbf{x}'),\,r\bigr)\) into \((F',R)\),
              with \(r\) counting satisfied constraints.
        \item[·] \emph{Balance (\(B\)).}  Controlled on the system and function registers 
              \((S,F)\) and \((S',F')\), encodes the acceptance amplitude
              \(\sqrt{A(\mathbf{x},\mathbf{x}')}\) onto the coin \(C\)
              via conditional rotations, where
              \(A=\min\{1,e^{-\beta\Delta f}\}\).
      \end{itemize}
\item \textbf{Shift (\(F\)).}  Performs a conditional swap
      \((S,F)\leftrightarrow(S',F')\), triggered only if the coin register \(C\) is in \(\ket{1}_C\) 
      and the constraint counter \(R\) is in \(\ket{m'}_R\). This effectively 
      accepts the proposed move if and only if it is feasible and statistically accepted.
\item \textbf{Reflection (\(R\)).}  Applies the phase flip given by the reflection              \(2\ket{0}\!\bra{0}_{S'}\!\otimes\!\ket{0}\!\bra{0}_{C}-I\).
\end{enumerate}

We now examine each operator individually to show how it is implemented.

\subsubsection{Candidate Generation (Proposal)}

The first stage of the Quantum Walk involves the application of the proposal operator $V$. This operator is responsible for generating a superposition of all potential candidates and evaluating their feasibility and cost. The constituent operations are schematically illustrated in Figure \ref{fig:V_operator}. 

The process begins by creating a uniform superposition over the entire search space (representing all representable integers with $d$ qubits):
\begin{equation}
    H^{\otimes d} \ket{0}_{S'}  = \frac{1}{\sqrt{2^d}} \sum_{k=-2^{d-1}}^{2^{d-1} - 1} \ket{k}_{S'}.
\end{equation}

Since not all generated states correspond to feasible solutions (i.e., some may violate constraints), we subsequently apply the constraint-counting subroutine defined in Section \ref{sec:restrictions}, Eq.~(\ref{eq:restrictions}). This procedure, illustrated in Figure \ref{fig:quantum_adder}, increments the register $R$ for each constraint satisfied by the state. Finally, we compute the objective function value, yielding the state:
\begin{equation}
         V\ket{0}_{S'} \ket{0}_{F'} \ket{0}_{R} = \frac{1}{\sqrt{|\Theta|}} \sum_{\mathbf{x}\in\Theta} \ket{\mathbf{x}}_{S'}  \ket{f(\mathbf{x})}_{F'} \ket{r(\mathbf{x})}_{R},
\end{equation}
where $r(\mathbf{x})$ denotes the number of constraints satisfied by the candidate $\mathbf{x}$. Given that this operator relies primarily on arithmetic addition circuits, the gate complexity—specifically the Toffoli count—is expected to scale linearly with the number of constraints \cite{Cuccaro:2004xxx}.

\subsubsection{Conditional Rotations}

Once the superposition of candidate moves has been generated by $V$, the acceptance probabilities are encoded into the coin register via the operator $B$, which is shown schematically in Figure \ref{fig:B_operator}. The first step involves subtracting the function value of the current state $S$ from that of the candidate state $S'$:

\begin{align}
\small
     \ket{\mathbf{x}}_{S} \ket{f(\mathbf{x})}_{F} &\ket{\mathbf{x}'}_{S'}  \ket{f(\mathbf{x}')}_{F'}  \mapsto \notag\\
    & \ket{\mathbf{x}}_{S} \ket{f(\mathbf{x})}_{F} \ket{\mathbf{x}'}_{S'}  \ket{f(\mathbf{x}') - f(\mathbf{x})}_{F'},
\end{align}
performing the subtraction using the same arithmetic logic described in Section \ref{sec:restrictions} for negative coefficients. The purpose of this step is to encode the acceptance probability as a quantum amplitude:

\begin{align}
    \small
    \ket{\mathbf{x}}_{S} &\ket{\mathbf{x}'}_{S'} \ket{0}_{C} \mapsto \notag\\
    & \ket{\mathbf{x}}_{S}\ket{\mathbf{x}'}_{S'}[\sqrt{1 - A(\mathbf{x}, \mathbf{x}')}\ket{0}_{C} + \sqrt{A(\mathbf{x}, \mathbf{x}')}\ket{1}_{C}] .
\end{align}
Here, the acceptance probability is encoded in the probability amplitude of the state $\ket{1}_C$, while the rejection probability corresponds to the amplitude of $\ket{0}_C$. Formally, this implements the rotation transformation:
\begin{equation} \label{eq:rotation}
\small
    R(\vartheta) =  \left(\begin{matrix}
        \sqrt{1-A(\mathbf{x}, \mathbf{x}')} &  -\sqrt{A(\mathbf{x}, \mathbf{x}')}\\
        \sqrt{A(\mathbf{x}, \mathbf{x}')} & \sqrt{1-A(\mathbf{x}, \mathbf{x}')}  
        \end{matrix} \right),
\end{equation}
defined by the rotation angle:
\begin{equation}
    \vartheta = \arcsin \left( \sqrt{A(\mathbf{x}, \mathbf{x}')}\right) = \arcsin \left( e^{-\frac{\beta}{2}[f(\mathbf{x}') - f(\mathbf{x})]}\right).
\end{equation}

To implement the conditional rotation shown in Fig.~\ref{fig:B_operator}, there is no need to compute the full exponential operator. Exponentiation constitutes a costly arithmetic operation in quantum circuits and can be avoided here, since it is not required for the optimization process (see Appendix~\ref{ap:spectral_decomposition}, Theorem~\ref{th:aprox_funct}). This avoids the heavy overhead of exponentiation and instead performs a qubit-wise interpolation between the binary-encoded components of the register $F'$. 

The Metropolis acceptance probability, defined in Eq.~(\ref{eq:acc_prob_overview}), satisfies $p_{\mathrm{acc}} = 1$ whenever the energy difference is non-positive ($\Delta f \leq 0$). Consequently, the non-trivial exponential decay need only be computed when the energy increases, i.e., for $\Delta f > 0$.

The register $F'$ encodes the cost function difference $\Delta f = f(\mathbf{x}') - f(\mathbf{x})$ using two's complement binary representation:

\begin{figure*}[t]
  \centering
  \includegraphics[width=0.7\textwidth]{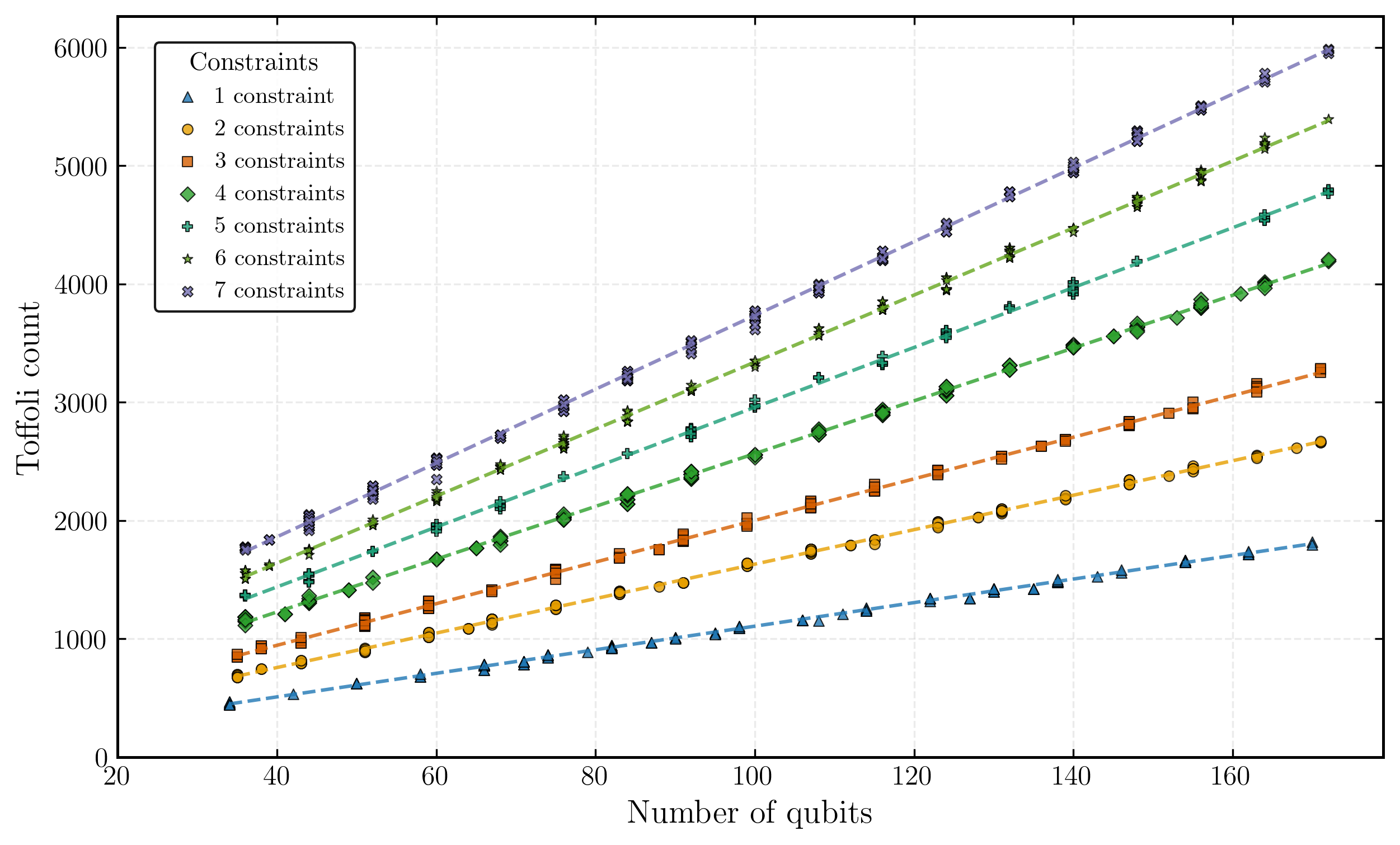}
  \caption{\justifying Toffoli-equivalent gate count versus total qubit count across 1000 randomly generated problem instances with $n=3$ variables, fixed inverse temperature $\beta=1$, and varying numbers of linear inequality constraints $m'$. Each color and marker shape corresponds to a distinct constraint count. The dashed lines show linear regression fits, confirming the predicted $\mathcal{O}(k)$ scaling with the total number of qubits $k$. The systematic increase in slope with $m'$ reflects the additional Toffoli overhead introduced by constraint-processing logic.}
  \label{fig:num_toffolis}
\end{figure*}

\begin{equation}
    y_{\mathrm{bin}} = b_{k-1} b_{k-2} \cdots b_{1} b_{0},
\end{equation}
where $b_i \in \{0,1\}$ denotes the $i$-th bit, and $b_{k-1}$ acts as the sign bit. In this encoding, $b_{k-1}=1$ indicates $\Delta f < 0$, while $b_{k-1}=0$ implies $\Delta f \geq 0$. Accordingly, the circuit logic is designed such that the exponential update branch is activated only when the sign bit is in the state $\ket{0}$ (representing positive energy differences), as illustrated in Fig.~\ref{fig:B_operator}.

To approximate the target rotation induced by the term $e^{-\frac{\beta}{2}\Delta f}$, we construct a linear mapping from the binary-encoded values in $F'$ to the rotation angles of the coin register. Instead of evaluating the nonlinear exponential, we define a linear interpolation that reproduces its behaviour over the discretized domain of $\Delta f$.  
Specifically, for $\Delta f > 0$, we adopt the linear model:
\begin{equation}
    e^{\frac{\beta}{2}\Delta f} \;\approx\;
    \lambda_0 b_0 + \lambda_1 b_1 + \cdots + \lambda_{k-2} b_{k-2},
\end{equation}
where $b_i \in \{0,1\}$ are the data qubits of $F'$, and $\{\lambda_i\}$ are real coefficients derived from an optimized linear fit. Crucially, these coefficients are constrained to guarantee monotonicity, a property sufficient to preserve the algorithm's thermalization behavior. Theorem~\ref{th:aprox_funct} in Appendix~\ref{ap:spectral_decomposition} provides the formal proof of validity. Note that this fit depends solely on the register size and the annealing schedule $\beta_k$, making it independent of the specific instance values or the functional form of the objective $f(\mathbf{x})$.

It is worth emphasizing that the purpose of this linearization is not to reproduce the exponential function with high numerical precision, but rather to preserve the correct thermal ordering of states according to their relative energies. In the Metropolis acceptance rule, the exponential factor $e^{-\beta \Delta f / 2}$ acts as a weighting term that suppresses transitions toward higher-energy configurations. Therefore, the algorithm only requires a monotonic mapping between $\Delta f$ and the rotation amplitudes in Eq.~(\ref{eq:rotation}) to ensure that $\Delta f_1 > \Delta f_2 \Rightarrow A(\mathbf{x}_1, \mathbf{x}'_1) < A(\mathbf{x}_2, \mathbf{x}'_2)$. Provided this monotonic relation holds, the random walk continues to bias the system toward the low-energy subspace. While the stationary distribution may deviate slightly from the exact Gibbs state due to the approximation, the ground-state support remains robust. The global convergence, governed by the spectral gap of $W_k$, is maintained under such smooth, monotonic reparameterizations.

After the rotation block is applied, the register $F'$ is uncomputed by applying $U_{\Delta f}^{\dagger}$, restoring it to $\ket{f(\mathbf{x}')}_{F'}$. The resulting transformation corresponds precisely to the conditional-rotation operator $B$, shown in Fig.~\ref{fig:B_operator}, which performs a qubit-wise linearized update of the coin conditioned on the sign and magnitude of $\Delta f$. 

From a resource perspective, the operator $B$ relies on two fundamental primitives: arithmetic subtractions and controlled rotations. The subtraction is implemented via ripple-carry adders, which incur a cost linear in the number of qubits. Similarly, by employing the linear approximation, we replace the costly exponentiation with a sequence of single-qubit controlled rotations, exactly one per qubit in the $F'$ register. Consequently, the total gate complexity is dominated by the Toffoli count, which scales linearly with the system size.

\subsubsection{Transition movements}\label{subsec:trans_movem}
At this stage, the conditional update of the state register $\ket{\mathbf{x}}_S$ is performed based on the acceptance and feasibility criteria:
\begingroup\footnotesize
\setlength{\jot}{2pt}
\begin{equation}
\begin{aligned}
F &\;\ket{\mathbf{x},f(\mathbf{x})}_{SF}\,\ket{\mathbf{x}',f(\mathbf{x}')}_{S'F'}\,\ket{\varphi}_C\ket{r}_R \\[2pt]
  &=
  \begin{cases}
    \ket{\mathbf{x}',f(\mathbf{x}')}_{SF}\,\ket{\mathbf{x},f(\mathbf{x})}_{S'F'} , 
      & \text{if } \ket{\varphi}_C\ket{r}_R = \ket{1}_C\ket{m'}_R,\\[2pt]
    \ket{\mathbf{x},f(\mathbf{x})}_{SF}\,\ket{\mathbf{x}',f(\mathbf{x}')}_{S'F'} ,
      & \text{otherwise.}
  \end{cases}
\end{aligned}
\end{equation}
\endgroup
for an ILP instance with $m'$ constraints. This operation is implemented via a sequence of Controlled-SWAP gates acting on the $S$ and $F$ registers, conditioned on the coin register $\ket{\varphi}_C$ and the constraint counter $\ket{r}_R$, as illustrated in Figure \ref{fig:algorithm_outline}.

Since a controlled-SWAP gate incurs a Toffoli cost that scales linearly with the number of control qubits, the constraint register plays a decisive role in the complexity of the operator $F$. Because the feasibility check is encoded through a cascade of $\lceil \log_{2} m'\rceil$ control qubits, where $m'$ is the number of constraints, the resulting multi-controlled SWAP determines the effective slope of the Toffoli count associated with this block. In other words, the constraint cardinality governs the dominant non-Clifford contribution of $F$, with a logarithmic scaling with the number of constraints.

\subsubsection{Reflection}
Finally, the reflection about the zero state $\ket{0}_{S'} \ket{0}_C$ is applied. This operator acts as follows:
\begin{align}
    R &\ket{\mathbf{x}'}_{S'} \ket{\varphi}_C  \notag \\
    &= \begin{cases}
         \ket{0}_{S'} \ket{0}_C & \text{if } \ket{\mathbf{x}'}_{S'} \ket{\varphi}_C = \ket{0}_{S'} \ket{0}_C, \\
         -\ket{\mathbf{x}'}_{S'} \ket{\varphi}_C & \text{otherwise.}
       \end{cases}
\end{align}
The operator is explicitly defined as $R = 2\ket{0}\bra{0}_{S'}  \otimes \ket{0}\bra{0}_C - I_{S'} \otimes I_C$, acting as the identity on all other registers. 

Physically, this transformation implements a phase flip conditioned on the subspace orthogonal to the zero state. Implementation-wise, it requires a multi-controlled operation that targets the entire $S'$ and $C$ registers (specifically, a zero-controlled phase shift). Since such a gate decomposes into a number of Toffoli-equivalent operations that scales linearly with the number of control qubits, the overall cost of the reflection step exhibits a clean linear dependence on the size of the state register.  

\begin{figure*}[t]
  \centering
  \includegraphics[width=0.7\textwidth]{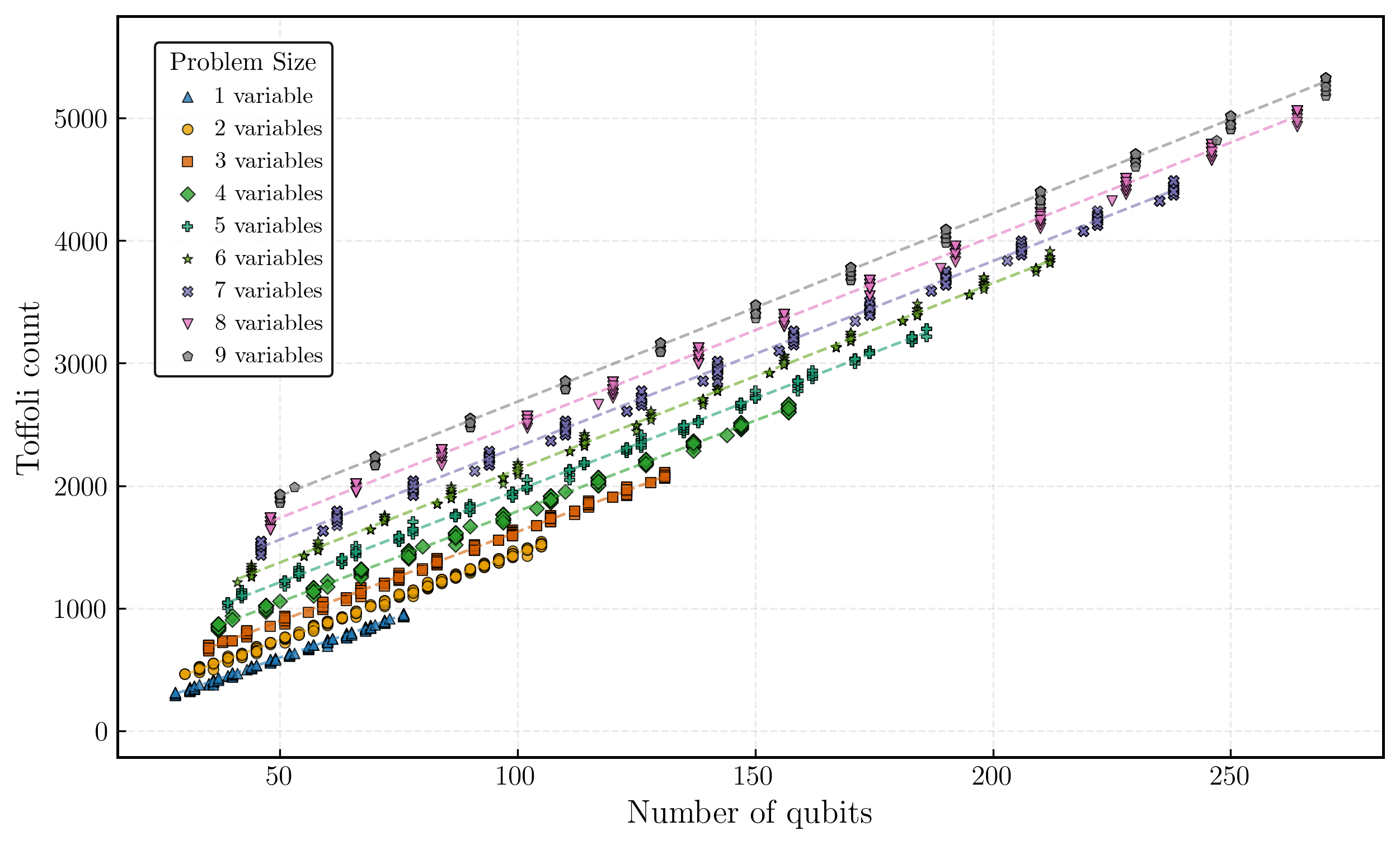}
    \caption{\justifying Toffoli gate counts versus total number of qubits for 1000 problem instances with $m'=2$ constraints of varying dimensionalities for the same constant $\beta = 1$. 
    Each color and marker shape corresponds to a different number of variables. 
    The dashed lines represent linear regression fits, showing a clear linear dependence of the Toffoli count on the total number of qubits with the same slope.}
    \label{fig:toffoli_vars}
\end{figure*}

Taken together, the four components of the walk operator—proposal ($V$), conditional rotation ($B$), transition ($F$), and reflection ($R$)—exhibit a uniform and highly structured resource behaviour. Each block is dominated by multi-controlled arithmetic or comparison routines, whose Toffoli-equivalent cost scales linearly with the number of qubits assigned to the corresponding registers. Thus, at the structural level, one expects the logical cost of a single Metropolis–Hastings update $W$ to grow proportionally to the overall circuit width.

A precise expression for the total number of qubits is shown directly from the register allocation in Table~\ref{tab:registers}. Let $d=\log_2 N$ be the number of qubits used to encode each variable representing the search space, $n$ the number of variables, and $m'$ the total number of constraints. The dominant contributions are:
\begingroup\footnotesize
\setlength{\jot}{4pt}
\begin{equation}\label{eq:aprox_num_qub}
Q_{\mathrm{total}}
= \mathcal{O}(n \log_2 N)+ \mathcal{O}(\log_2 (N\sum_k |c_k|) )+\mathcal{O}(\log_2(m'))
\end{equation}
\endgroup
Hence, $Q_{\mathrm{total}} = \Theta(n\log_2 N)$ with only logarithmic corrections in the number of constraints and small additive terms from the arithmetic registers.  

Because the search interval grows exponentially with $d$ ($x_i\in[-2^{d-1},\,2^{d-1}-1]$), this implies that the total qubit requirement scales only logarithmically with the size of the discrete search space.  

In summary, the walk operator $W$ exhibits a structurally linear dependence on the encoded ILP dimension—both in qubits and in Toffoli-equivalent gates.  Section~\ref{sec:results} confirms this behaviour numerically across thousands of randomly generated ILP instances.

\section{Simulation Results} \label{sec:results}

Crucially, the results presented in this section rely on a complete, end-to-end implementation of the proposed algorithm, validated via gate-level simulation using standard circuit synthesis tools. By explicitly programming the full quantum circuit—rather than relying solely on analytical bounds or high-level functional simulations—we are able to conduct a global assessment that validates two complementary aspects simultaneously: physical resource scaling and algorithmic convergence fidelity.

\subsection{Problem Size Scaling}

\subsubsection{Metrics and simulation setup}\label{subsec:setup}

While an analytical asymptotic analysis provides valuable insight into the expected scaling behavior, a precise and implementation-independent measure can only be obtained by explicitly counting the gates required to realize the quantum Metropolis operator $W$. Therefore, in this section we adopt a numerical approach: the full algorithm has been implemented, and the exact number of gates in each constituent operator of the walk $W$ is extracted to empirically characterize the resource cost.

To ensure a fair comparison between different configurations, in this paper we quantify the $T$-gate cost primarily at the level of Toffoli-equivalent gates. This metric is motivated by the fact that Clifford operations (such as Hadamard or CNOT) are relatively simple to implement and exhibit low physical error rates, whereas non-Clifford gates, particularly multi-controlled operations, dominate the fault-tolerant resource cost in realistic quantum architectures \cite{Wang2025, Ortega2025}. Therefore, expressing the total gate count in Toffoli-equivalents provides a more meaningful indicator of the true computational overhead.

\begin{figure*}[t]
  \centering
  \includegraphics[width=0.65\textwidth]{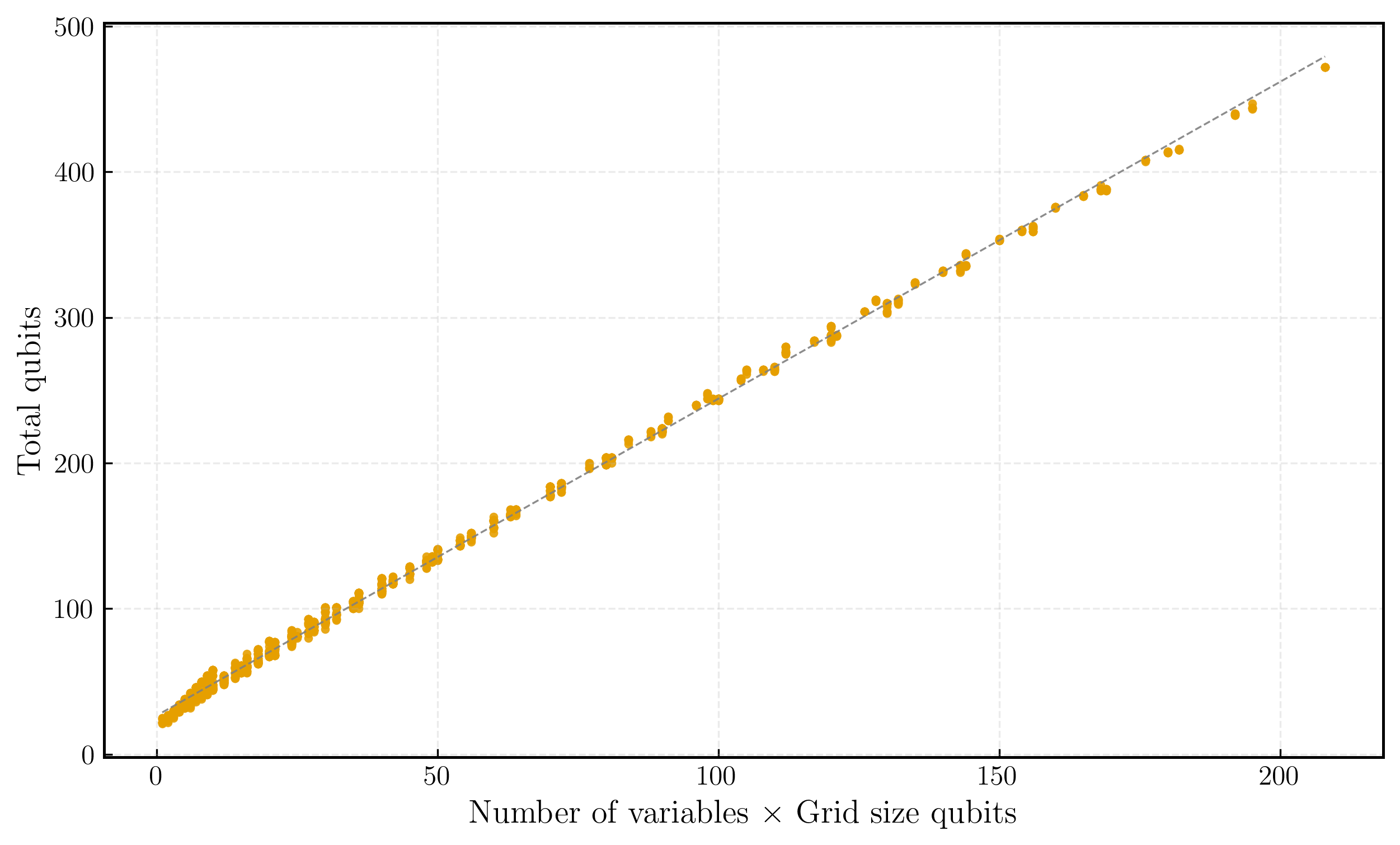}
  \caption{\justifying Total qubit count versus the product $n d$ (number of variables times discretization qubits per variable) across 1500 randomly generated problem instances. Each variable is encoded using a $d$-qubit two’s-complement representation. The dashed line shows a linear regression fit ($R^2>0.99$), confirming the predicted $\mathcal{O}(nd)$ scaling dominated by the system and proposal registers $S$ and $S'$. The fitted slope also incorporates additive contributions from the auxiliary registers $(F,F',R,C)$, illustrating a predictable and nearly affine growth of the total qubit requirement with the encoded problem size.}
  \label{fig:num_qubits}
\end{figure*}

In this work, the following Toffoli-equivalent cost model is adopted \cite{Nielsen_Chuang_2010, Ortega2025, RevModPhys.74.347}. Each multi-controlled X gate (MCX) with $n_c$ control qubits contributes  
\begin{equation}
C_{\mathrm{MCX}} = 2n_c - 3,
\end{equation}
Toffoli-equivalents.  

Controlled unitaries (CU) follow a similar scaling, as their non-Clifford cost increases linearly with the number of control qubits \cite{Ortega2025}:
\begin{equation}
C_{\mathrm{CU}} = 2n_c - 2.
\end{equation}

Controlled-SWAP gates are automatically decomposed into MCX operations during synthesis; consequently, their resource footprint is captured directly via the total MCX count and depth. By definition, a standard Toffoli (CCX) gate corresponds to exactly one Toffoli-equivalent.

Regarding single-qubit rotations, the cumulative contribution of $T$ and $T^\dagger$ gates is estimated based on a standard distillation overhead of one Toffoli-equivalent per seven $T$-type gates:
\begin{equation}
C_{T} = \frac{N_T + N_{T^\dagger}}{7},
\end{equation}
where $N_T$ and $N_{T^\dagger}$ denote the total counts of $T$ and $T^\dagger$ gates, respectively.

All remaining Clifford gates, including CNOT and Hadamard are considered to have negligible Toffoli cost. 

Each multi-controlled operation is thus mapped to its Toffoli-equivalent contribution, capturing the dominant non-Clifford component of the total circuit depth and logical resource requirements. This counting scheme provides a unified and hardware-agnostic measure of logical complexity, independent of the specific gate decomposition or hardware connectivity.

To apply this cost model, the full algorithm is implemented using the Qiskit framework \cite{qiskit2024}. We utilize standard synthesis routines to decompose the high-level circuit components into the specific fundamental operations defined above—including MCX gates, single-qubit $T/T^\dagger$ rotations, and Clifford primitives. This breakdown allows for a direct mapping between the synthesized circuit topology and the Toffoli-equivalent metrics, enabling a precise and deterministic calculation of the total resource overhead.

We consider families of integer linear programs of the form in Eq.~(\ref{eq:canonicalLP}), with \( n \in [1, 9] \) variables and \( m' \in [1, 7] \) linear constraints. Coefficients \( c_i \) and constraint entries \( g_i, h_j \) are randomly sampled from bounded integer intervals, transforming the equality constraints into two inequality constraints, and each variable is represented using \( d \)-qubit two’s-complement encoding. For each instance, we record the total number of qubits required and the equivalent Toffoli count. Among the possible annealing steps, we focus specifically on the first iteration of the quantum Metropolis operator \( W \), which captures the fundamental cost and scaling characteristics of the algorithm under identical thermalization conditions, with a fixed inverse temperature \(\beta = 1\). A more detailed study of the effect of the annealing schedule and the number of iterations on convergence and accuracy is presented in Sec.~\ref{subsec:precision}.

\subsubsection{Simulation results}

\begin{figure*}[t]
  \centering
  \includegraphics[width=\textwidth]{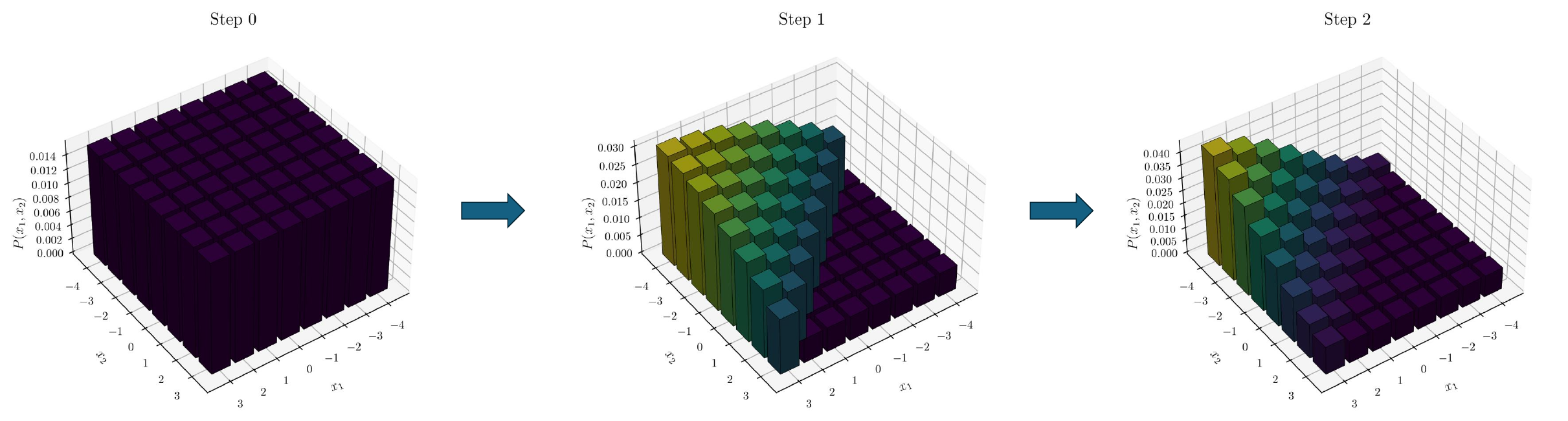}
    \caption{\justifying Histogram of the measured probability distribution obtained after applying the \textit{Randomized Quantum Metropolis} algorithm for two variables and three discretization qubits per variable. 
    Each bar represents the probability of measuring a specific computational basis state $\ket{x_1, x_2}$, corresponding to a discrete point in the feasible region defined by the constraint $-x_1 + x_2 \ge 0$. 
    The optimization objective is $f(x_1, x_2) = -2x_1 - 2x_2$, and the resulting distribution illustrates the progressive thermalization process toward the global minimum. 
    With three qubits per variable, the search domain spans the square $[-4,3]\times[-4,3]$.}
  \label{fig:3_qubits_inference}
\end{figure*}

Figure~\ref{fig:num_toffolis} illustrates the scaling of the Toffoli-equivalent gate count as a function of the total qubit register size. These results, derived from 1000 random instances per configuration at $\beta = 1$, reveal a strictly linear relationship between logical depth and circuit width across all regimes.

Crucially, while the scaling remains linear, the slope of the regression fit steepens systematically with the number of linear constraints $m'$. This behavior directly reflects the resource overhead associated with the constraint-checking logic—specifically, the multi-controlled operations required by the conditional update operator $F$ (as detailed in Sec.~\ref{subsec:trans_movem}). Consequently, the data confirms that increasing the constraint density contributes additively to the logical cost coefficient, without altering the fundamental linear scaling of the one step operator $W$.

Complementing this analysis, Fig.~\ref{fig:toffoli_vars} isolates the effect of the number of variables by fixing the constraint count and varying the problem dimensionality. In this regime, the Toffoli-equivalent count continues to exhibit a linear dependence on the total number of qubits, with a slope that remains essentially invariant across different values of $n$. 
Increasing the number of variables results in a vertical upward shift of the regression lines—indicating a larger baseline circuit cost—but does not alter the asymptotic scaling rate. 
This behavior indicates that the number of variables affects the offset of the resource count, while the dominant scaling is governed by the total qubit budget rather than by $n$ alone.

Figure~\ref{fig:num_qubits} addresses the spatial complexity of the construction by plotting the total number of qubits as a function of the product between the number of variables $n$ and the number of discretization qubits $d$ used to encode each variable in two’s-complement form. The data exhibit a linear dependence, confirming that the total qubit requirement scales proportionally with the effective encoded problem size $n d = n \log_2 N$, where $N$ denotes the size of the discretized domain per variable. Minor deviations from ideal linearity arise from secondary structural contributions, such as ancillary registers used for constraint evaluation and the bit-width required to store intermediate arithmetic results in the $F'$ register. These effects introduce only mild dispersion and do not modify the underlying linear trend, in agreement with the analytical register estimates summarized in Table~\ref{tab:registers} and Eq.~(\ref{eq:aprox_num_qub}).

Taken together, Figs.~\ref{fig:num_toffolis}, \ref{fig:toffoli_vars}, and~\ref{fig:num_qubits} establish that both the logical depth (measured in Toffoli-equivalent gates) and the spatial complexity (measured in total qubits) of the quantum Metropolis operator $W$ scale linearly with the encoded problem size. 
While the classical search space grows exponentially with the discretization precision as $[-2^{d-1},\,2^{d-1}-1]$, the required quantum resources scale only as $\mathcal{O}(n d)$. 
This separation between exponential classical complexity and polynomial quantum resource growth highlights a key advantage of the proposed formulation: it enables coherent exploration of an exponentially large feasible region while maintaining predictable and scalable circuit complexity.

Error bars are omitted from the reported figures, as the objective of the analysis is to reveal asymptotic scaling behavior rather than statistical variability. The quantities plotted correspond to deterministic circuit constructions, and the observed trends directly reflect algorithmic complexity bounds rather than stochastic fluctuations.

\subsection{Precision and convergence analysis} \label{subsec:precision}

In this section we move beyond the analysis of individual operators and simulate the complete algorithm described in Sec.~\ref{sec:overview}. 

The annealing schedule implemented in our simulations starts from a fully ``hot'' configuration at $\beta = 0$ and gradually cools toward $\beta \rightarrow \infty$. 
This is achieved by combining a linear increase of $\beta$ in the early stages ($0 \le \beta \le 1$) with an exponential decay of the acceptance amplitude. 
Each thermal block consists of a randomly selected number of internal repetitions, $T \in [1,3]$, corresponding effectively to $q=2$ qubits of randomization depth in the quantum walk.  

Intuitively, when consecutive $\beta$ values are closer—i.e., when the temperature decreases more smoothly—the system requires fewer randomization steps to maintain detailed balance between successive Gibbs-like states.  
However, this also implies that more distinct $\beta_k$ stages (larger $Q$) must be executed to reach the same final temperature.  
Conversely, a faster cooling schedule reduces the total number of annealing stages but requires more mixing per stage to achieve comparable convergence.  
The chosen schedule thus represents a practical heuristic trade-off between coherence, convergence stability, and total circuit depth, tuned to each problem’s energy landscape and constraint configuration.

Figure~\ref{fig:3_step_inference} illustrates the evolution of the probability distribution throughout the annealing process for a representative instance of the integer linear program. The depicted case corresponds to an instance with two optimization variables, each encoded on two qubits in two’s-complement form, defining a discrete search space of $[-2, 1] \times [-2, 1]$. 
In subsequent experiments shown Fig. ~\ref{fig:3_qubits_inference}, we increase the discretization to three qubits per variable, expanding the hypercubic search domain to $[-4, 3]\times[-4, 3]$ and thereby raising the maximum qubits of the qiskit simulators.

The algorithm employs a $\beta$-scheduling scheme that gradually increases the inverse temperature, effectively performing a quantum analogue of simulated annealing. As $\beta$ grows, the distribution concentrates progressively around configurations with lower energy values, until it converges with high probability to the global or local minimum of the cost landscape. 
This thermalization process is fully coherent and reversible, while partial measurements at selected steps remove entanglement with ancillary registers, effectively purging computational “garbage” and preserving only the physically relevant amplitudes.

A distinctive feature of this quantum Metropolis formulation is that, unlike amplitude-amplification schemes such as Grover’s algorithm, the probability of measuring the optimal state does not exhibit oscillatory behavior as a function of the number of iterations. In practice, the repeated application of distinct operators $W_k$, each corresponding to a progressively lower effective temperature, combined with the partial measurements of the coin and movement registers, leads to a steady concentration of probability around low-energy configurations. 
While a strict proof of monotonicity is beyond the present scope, the numerical simulations consistently show a non-oscillatory, cumulative increase in the occupation probability of near-optimal states. 
This behavior can be interpreted as a consequence of measurement-induced state purification together with the $\beta$-dependent cooling schedule, which together suppress transitions that would otherwise repopulate higher-energy configurations. 
Hence, even though the optimal number of $W_k$ applications cannot be fixed a priori—since it depends on the specific cost landscape and annealing path—the observed convergence remains stable and qualitatively monotonic across all tested instances. The key advantage is that, in the quantum case, the process exhibits a quadratic gain in convergence rate with respect to classical random-walk mixing \cite{PhysRevLett.101.130504}, while retaining the same qualitative monotonic behavior of probability concentration.

In summary, the results shown in Fig.~\ref{fig:3_step_inference}  and Fig. ~\ref{fig:3_qubits_inference} demonstrate how the proposed quantum Metropolis dynamics progressively amplifies the occupation probability of the optimal solution as annealing proceeds, evidencing both the correctness of the $\beta$-dependent construction and the robustness of the partial measurement mechanism in maintaining a physically consistent thermal evolution.

\section{Conclusions} \label{sec:conclusions}
This work introduces the first fully quantum Metropolis--Hastings algorithm for integer linear programming, implemented entirely as reversible quantum circuits without quantum-RAM assumptions or classical pre/post-processing. 
This represents a genuine advance in quantum constraint programming: all arithmetic—objective evaluation, constraint verification, acceptance filtering—occurs coherently within the quantum circuit, enabling simultaneous evaluation of candidate solutions across the entire feasible polytope in superposition. 
This capability is fundamentally unavailable in hybrid approaches (e.g., quantum oracles with classical search) or variational methods (QAOA, VQE) that require repeated quantum--classical feedback loops.

A central contribution of this study is the explicit resource characterization: one step of the algorithm requires $\mathcal{O}(n \log_2 N)$ logical qubits and $\mathcal{O}(k)$ Toffoli-equivalent gates, where $n$ is the number of variables, $N$ the search interval, and $k$ the total qubit count. 
This linear scaling with problem dimension—despite the exponential size of the classical search space—is rigorously derived and validated numerically across more than 1500 instances (see Figs.~\ref{fig:num_toffolis}, \ref{fig:toffoli_vars}, and \ref{fig:num_qubits}). 
Crucially, while the total time to convergence depends on the spectral gap—reflecting the NP-hard nature of the problem—the \textit{implementation cost per step} is guaranteed to remain polynomial. 
Unlike classical branch-and-bound methods where worst-case memory and branching costs can become intractable, our framework offers predictable, linear scaling of the quantum hardware resources required to execute the Metropolis operator.

Beyond finding the global optimum, the algorithm naturally generates a structured ranking of near-optimal solutions weighted by their Boltzmann probability. 
This capability is particularly valuable in practical applications such as scheduling, resource allocation, and design optimization, where high-quality feasible solutions are often preferred over exhaustive optimization. 
This behavior distinguishes our approach from both classical heuristics and quantum approximate solvers, which typically target a single best solution.

Numerical simulations confirm convergence toward the global minimum with high probability (see Fig.~\ref{fig:3_step_inference} and Fig.~\ref{fig:3_qubits_inference}) and validate the predicted linear resource scaling (Figs.~\ref{fig:num_toffolis}--\ref{fig:num_qubits}). 
The algorithm's robustness stems from its fully reversible construction and the thermalizing dynamics of the quantum Metropolis walk, which progressively concentrate probability amplitude on low-cost feasible configurations without the oscillation risks inherent in amplitude amplification schemes.

Critically, this work addresses the theory--practice gap in quantum optimization by providing a framework whose resource requirements are characterized at the logical level and whose implementation is amenable to systematic error correction. 
While the method assumes access to fault-tolerant quantum computers, this transparent resource accounting enables principled technology roadmapping: practitioners can estimate physical qubit overheads based on realistic error rates and plan deployment timelines accordingly.

Future directions include: (i) integration of problem-specific proposal mechanisms to exploit structure in integer linear programming instances, (ii) extension to non-linear and mixed-integer nonlinear programming, and (iii) experimental validation on near-term quantum processors using error mitigation techniques. 
Overall, the proposed framework opens a clear pathway toward practical quantum advantage in combinatorial optimization as quantum hardware matures.

\section{Acknowledgements}
G.E., R.C., and M.A.M.-D. acknowledge the support from grants MINECO/FEDER Projects, PID2021-122547NB-I00 FIS2021, MADQuantumCM project funded by Comunidad de Madrid, the Recovery, Transformation, and Resilience Plan, NextGenerationEU, funded by the European Union, and the Ministry of Economic Affairs Quantum ENIA project funded by Madrid ELLIS Unit CAM. G. E. also acknowledge the support from the CAM Program TEC-2024/COM-84 QUITEMAD-CM.
M.A.M.-D. has also been partially supported by the U.S. Army Research Office through Grant No.W911NF-14-1-0103. This work has been financially supported by the Ministry for Digital Transformation and of Civil Service of the Spanish Government through the QUANTUM ENIA project call – Quantum Spain project, and by the European Union through the Recovery, Transformation and Resilience Plan – NextGenerationEU within the framework of the Digital Spain 2026 Agenda.

\newpage
\appendix

\bibliography{bibliography}

\onecolumngrid
\clearpage
\makeatletter

\setcounter{equation}{0}
\setcounter{figure}{0}
\setcounter{table}{0}
\setcounter{page}{1}
\setcounter{section}{0}
\renewcommand{\theequation}{\thesection.\arabic{equation}}
\renewcommand{\thefigure}{\arabic{figure}}
\renewcommand{\thetable}{\arabic{table}}

\section{Spectral Decomposition of the Walk Operator} \label{ap:spectral_decomposition}

\begin{lemma} \label{w_unitary}
    The construction $W \;=\; R\,P^{\dagger}\,F\, P$ is unitary.
\end{lemma}

\begin{proof}
    Although the unitarity of the operator $W$ is evident because it has been constructed from unitary operations, it is straightforward to show that $W^{\dagger} W = I$ for a general state $\sum_{x \in \Omega} a_x \ket{x}_S \ket{f(x)}_F \ket{0}_{S'}\ket{0}_{F'} \ket{0}_R  \ket{c}_C$. Since registers $\ket{0}_{S'}\ket{0}_{F'} \ket{0}_R $ are ancillas, at the end of an entire step of the quantum walk, they must always return to zero, making this the most general state. 

    We can expand the operator:
    \begin{equation}
        W^{\dagger} W =(R V^{\dagger} B^{\dagger} F B V)^{\dagger} R V^{\dagger} B^{\dagger} F B V = V^{\dagger} B^{\dagger} F^{\dagger} B V R^{\dagger} R V^{\dagger} B^{\dagger} F B V
    \end{equation}

    Since $R$ is nothing more than a reflection, it is easy to verify that:
    \begin{equation}
        R^\dagger R = (2\ket{0}\bra{0}_{S'}  \otimes \ket{0}\bra{0}_C - I_{S'} \otimes I_C)^\dagger (2\ket{0}\bra{0}_{S'}  \otimes \ket{0}\bra{0}_C - I_{S'} \otimes I_C) = I_{S'} \otimes I_C,
    \end{equation}
    acting as the identity in the rest of the registers. Based on how we have constructed the operator $V$, we can see that:
        \begin{align}
    \small
        & V V^\dagger = U_f\prod_{i=0}^{p - 1} G_{p - 1 -i} \prod_{j=0}^{m - 1} H_{m - 1 -j} U_H^{\otimes n}(U_f\prod_{i=0}^{p - 1} G_{p - 1 -i} \prod_{j=0}^{m - 1} H_{m - 1 -j} U_H^{\otimes n})^\dagger \nonumber \\
        & = U_f\prod_{i=0}^{p - 1} G_{p - 1 -i} \prod_{j=0}^{m - 1} H_{m - 1 -j} U_H^{\otimes n}(U_H^{\otimes n})^{\dagger}  (\prod_{j=0}^{m - 1} H_{m - 1 -j})^{\dagger}(\prod_{i=0}^{p - 1} G_{p - 1 -i})^{\dagger} U_f^{\dagger} \nonumber \\
        & = U_f\prod_{i=0}^{p - 1} G_{p - 1 -i} \prod_{j=0}^{m - 1} H_{m - 1 -j}  (\prod_{j=0}^{m - 1} H_{m - 1 -j})^{\dagger}(\prod_{i=0}^{p - 1} G_{p - 1 -i})^{\dagger} U_f^{\dagger}.
    \end{align}
    Each of the operators $H_j, G_i$ has been constructed in a unitary way, since:
    \begin{equation}
        H_j^\dagger H_j = (U_{h_j}^{\dagger} A_{\text{eq}} U_{h_j})^\dagger U_{h_j}^{\dagger} A_{\text{eq}} U_{h_j} = U_{h_j}^{\dagger} A_{\text{eq}}^{\dagger} U_{h_j} U_{h_j}^{\dagger} A_{\text{eq}} U_{h_j},
    \end{equation}
    and similarly
    \begin{equation}
        G_i^\dagger G_i = (U_{g_i}^{\dagger} A_{\text{ineq}} U_{g_i})^\dagger U_{g_i}^{\dagger} A_{\text{ineq}} U_{g_i} = U_{g_i}^{\dagger} A_{\text{ineq}}^{\dagger} U_{g_i} U_{g_i}^{\dagger} A_{\text{ineq}} U_{g_i}.
    \end{equation}
    The operators $U_{h_j}$, $U_{g_i}$, and $U_f$ implement only reversible arithmetic transformations and thus they are unitary by construction. 
    Similarly, the operators $A_{\mathrm{eq}}$ and $A_{\mathrm{ineq}}$ consist solely of conditional addition circuits, which are also reversible. This leads to $V V^\dagger = I$.
    
    Repeating the process for operator $B$:
    \begin{align}
    \small
        & B^\dagger B = (U_{\Delta f}^\dagger R(\vartheta) U_{\Delta f})^\dagger U_{\Delta f}^\dagger   R(\vartheta)  U_{\Delta f} \nonumber \\
        & =  U_{\Delta f}^\dagger   R(\vartheta)^\dagger  U_{\Delta f} U_{\Delta f}^\dagger   R(\vartheta)  U_{\Delta f}.
    \end{align}
    The operator $U_{\Delta f}$ is unitary by construction, since it is quantum arithmetic, and the operator $R(\vartheta)$ is constructed from conditioned rotations, so it is also unitary, therefore $B^\dagger B = I$.

    Finally, the operator $F$ is a controlled swap, so trivially $F^\dagger F = I$. Thus, since all operators that constitute the operator $W$ are unitary, it follows that operator $W$ is unitary.
\end{proof}

The properties described here have been demonstrated for quantum Metropolis operators \cite{PhysRevResearch.5.033059}, but it is necessary to demonstrate them for this explicit construction.

\begin{theorem}\label{eigen_W}
The state $\ket{\Pi} = \frac{1}{\sqrt{Z(\beta)}}\sum_{x \in \Omega} \sqrt{e^{- \beta f(x)}} \ket{x}_S \ket{f(x)}_F \ket{0}_{S'}\ket{0}_{F'} \ket{0}_R  \ket{0}_C$, with $Z(\beta) = \sum_{x \in \Omega} e^{-\beta f(x)}$ is eigenstate of $W \;=\; R\,P^{\dagger}\,F\, P $ with eigenvalue 1. 
\end{theorem}
\begin{proof}
    To prove this theorem, we will perform a straightforward calculation. Let $m'$ be the total number of constraints. Let's break down $P = BV$ and apply all the operators in order to $ \ket{\Pi}$. Introducing $\mathcal{N}=Z(\beta)|\Theta|$:
\begin{align}
\small
& R V^{\dagger} B^{\dagger} F B V 
  \frac{1}{\sqrt{Z(\beta)}}
  \sum_{x \in \Omega} 
  \sqrt{e^{-\beta f(x)}} 
  \ket{x}_S \ket{f(x)}_F 
  \ket{0}_{S'} \ket{0}_{F'} \ket{0}_R \ket{0}_C 
  \nonumber \\[4pt]
&= R V^{\dagger} B^{\dagger} F B 
  \frac{1}{\sqrt{\mathcal{N}}} 
  \sum_{x \in \Omega} 
  \sqrt{e^{-\beta f(x)}} 
  \ket{x}_S \ket{f(x)}_F 
  \sum_{x' \in \Theta} 
  \ket{x'}_{S'} \ket{f(x')}_{F'} 
  \ket{r_{x'}}_R \ket{0}_C 
  \nonumber \\[4pt]
&= R V^{\dagger} B^{\dagger} F 
  \frac{1}{\sqrt{\mathcal{N}}} 
  \sum_{x \in \Omega} \sum_{x' \in \Theta} 
  \sqrt{e^{-\beta f(x)}} 
  \ket{x}_S \ket{f(x)}_F 
  \ket{x'}_{S'} \ket{f(x')}_{F'} 
  \ket{r_{x'}}_R 
  \big[ 
    \sqrt{1-A(x,x')} \ket{0}_C 
    + \sqrt{A(x,x')} \ket{1}_C 
  \big]
  \nonumber \\[4pt]
&= R V^{\dagger} B^{\dagger} 
  \frac{1}{\sqrt{\mathcal{N}}} 
  \bigg[ 
    \sum_{x \in \Omega} \sum_{x' \in \Omega} 
    \sqrt{e^{-\beta f(x)}} 
    \ket{x}_S \ket{f(x)}_F 
    \ket{x'}_{S'} \ket{f(x')}_{F'} 
    \ket{m'}_R 
    \sqrt{1-A(x,x')} \ket{0}_C 
  \nonumber \\[-2pt]
&\qquad\qquad 
  + \sum_{x \in \Omega} \sum_{x' \in \Omega} 
    \sqrt{e^{-\beta f(x)}} 
    \ket{x'}_S \ket{f(x')}_F 
    \ket{x}_{S'} \ket{f(x)}_{F'} 
    \ket{m'}_R 
    \sqrt{A(x,x')} \ket{1}_C 
  \bigg]
  \nonumber \\[4pt]
&\quad + R V^{\dagger} B^{\dagger} 
  \frac{1}{\sqrt{\mathcal{N}}} 
  \sum_{x \in \Omega} \sum_{x' \in \Theta - \Omega} 
  \sqrt{e^{-\beta f(x)}} 
  \ket{x}_S \ket{f(x)}_F 
  \ket{x'}_{S'} \ket{f(x')}_{F'} 
  \ket{r_{x'}}_R 
  \big[ 
    \sqrt{1-A(x,x')} \ket{0}_C 
    + \sqrt{A(x,x')} \ket{1}_C 
  \big].
\end{align}

    Let's examine the three terms of the sum separately. Applying $B^{\dagger}$ to the first one:
    \begin{align}
    \small
    & R V^{\dagger} B^{\dagger} \frac{1}{\sqrt{\mathcal{N}}} 
      \Big[ \sum_{x \in \Omega} \sum_{x' \in \Omega} 
      \sqrt{e^{- \beta f(x)}} \ket{x}_S \ket{f(x)}_F 
      \ket{x'}_{S'} \ket{f(x')}_{F'} \ket{m'}_R 
      \sqrt{1-A(x,x')} \ket{0}_C \Big] \nonumber \\
    & = R V^{\dagger}  \frac{1}{\sqrt{\mathcal{N}}} 
      \Big[ \sum_{x \in \Omega} \sum_{x' \in \Omega} 
      \sqrt{e^{- \beta f(x)}} \ket{x}_S \ket{f(x)}_F 
      \ket{x'}_{S'} \ket{f(x')}_{F'} \ket{m'}_R 
      \nonumber\\
    & \qquad\qquad \qquad 
      \sqrt{1-A(x,x')} 
      \big[ \sqrt{1-A(x,x')} \ket{0}_C 
      - \sqrt{A(x,x')} \ket{1}_C \big] \Big] .
\end{align}
    Applying $B^{\dagger}$ to the second:
\begin{align}
\small
& R V^{\dagger} B^{\dagger} 
   \frac{1}{\sqrt{\mathcal{N}}}\bigg[ 
    \sum_{x \in \Omega} \sum_{x' \in \Omega} 
    \sqrt{e^{-\beta f(x)}} 
    \ket{x'}_S \ket{f(x')}_F 
    \ket{x}_{S'} \ket{f(x)}_{F'} 
    \ket{m'}_R 
    \sqrt{A(x,x')} \ket{1}_C 
  \bigg]
  \nonumber \\[4pt]
&= R V^{\dagger}  
  \frac{1}{\sqrt{\mathcal{N}}} \bigg[ 
    \sum_{x \in \Omega} \sum_{x' \in \Omega} 
    \sqrt{e^{-\beta f(x)}} 
    \ket{x'}_S \ket{f(x')}_F 
    \ket{x}_{S'} \ket{f(x)}_{F'} 
    \ket{m'}_R 
    \sqrt{A(x,x')} 
    \big[ 
      \sqrt{A(x',x)} \ket{0}_C   \nonumber \\[4pt]
& \quad\quad\quad\quad\quad\quad\quad\quad\quad\quad\quad\quad\quad\quad\quad\quad \quad\quad\quad\quad\quad\quad      + \sqrt{1-A(x',x)} \ket{1}_C 
    \big] 
  \bigg].
\end{align}

    Using the detailed balance condition:
    \begin{equation}
        \sqrt{e^{- \beta f(x)}}A(x,x')=\sqrt{e^{- \beta f(x')}}A(x',x),
    \end{equation}
    And swapping the dummy indices $ x\leftrightarrow x'$:
\begin{align}
\small
& R V^{\dagger}  
  \frac{1}{\sqrt{\mathcal{N}}}\bigg[ 
    \sum_{x \in \Omega} \sum_{x' \in \Omega} 
    \sqrt{e^{-\beta f(x)}} 
    \ket{x'}_S \ket{f(x')}_F 
    \ket{x}_{S'} \ket{f(x)}_{F'} 
    \ket{m'}_R 
    \sqrt{A(x,x')} 
    \big[ 
      \sqrt{A(x',x)} \ket{0}_C 
      + \sqrt{1-A(x',x)} \ket{1}_C 
    \big] 
  \bigg]
  \nonumber \\[4pt]
&= R V^{\dagger}  
  \frac{1}{\sqrt{\mathcal{N}}}\bigg[ 
    \sum_{x \in \Omega} \sum_{x' \in \Omega} 
    \sqrt{e^{-\beta f(x')}} 
    \ket{x'}_S \ket{f(x')}_F 
    \ket{x}_{S'} \ket{f(x)}_{F'} 
    \ket{m'}_R 
    \sqrt{A(x',x)} 
    \big[ 
      \sqrt{A(x',x)} \ket{0}_C 
      + \sqrt{1-A(x',x)} \ket{1}_C 
    \big] 
  \bigg]
  \nonumber \\[4pt]
&= R V^{\dagger}  
  \frac{1}{\sqrt{\mathcal{N}}}\bigg[ 
    \sum_{x' \in \Omega} \sum_{x \in \Omega} 
    \sqrt{e^{-\beta f(x)}} 
    \ket{x}_S \ket{f(x)}_F 
    \ket{x'}_{S'} \ket{f(x')}_{F'} 
    \ket{m'}_R 
    \sqrt{A(x,x')} 
    \big[ 
      \sqrt{A(x,x')} \ket{0}_C 
      + \sqrt{1-A(x,x')} \ket{1}_C 
    \big] 
  \bigg].
\end{align}

    Applying $B^{\dagger}$ to the third:
\begin{align}
\small
& R V^{\dagger} B^{\dagger} 
  \frac{1}{\sqrt{\mathcal{N}}}\bigg[
    \sum_{x \in \Omega} \sum_{x' \in \Theta - \Omega} 
    \sqrt{e^{-\beta f(x)}} 
    \ket{x}_S \ket{f(x)}_F 
    \ket{x'}_{S'} \ket{f(x')}_{F'} 
    \ket{r_{x'}}_R 
    \big[
      \sqrt{1-A(x,x')} \ket{0}_C 
      + \sqrt{A(x,x')} \ket{1}_C
    \big]
  \bigg]
  \nonumber \\[4pt]
&= R V^{\dagger}  
  \frac{1}{\sqrt{\mathcal{N}}}\bigg[
    \sum_{x \in \Omega} \sum_{x' \in \Theta - \Omega} 
    \sqrt{e^{-\beta f(x)}} 
    \ket{x}_S \ket{f(x)}_F 
    \ket{x'}_{S'} \ket{f(x')}_{F'} 
    \ket{r_{x'}}_R 
    \sqrt{1-A(x,x')} 
    \big[
      \sqrt{1-A(x,x')} \ket{0}_C 
      - \sqrt{A(x,x')} \ket{1}_C
    \big]
  \nonumber \\[-2pt]
&\qquad\qquad
  + \sum_{x \in \Omega} \sum_{x' \in \Theta - \Omega} 
    \sqrt{e^{-\beta f(x)}} 
    \ket{x}_S \ket{f(x)}_F 
    \ket{x'}_{S'} \ket{f(x')}_{F'} 
    \ket{r_{x'}}_R 
    \sqrt{A(x,x')} 
    \big[
      \sqrt{A(x,x')} \ket{0}_C 
      - \sqrt{1-A(x,x')} \ket{1}_C
    \big]
  \bigg]
  \nonumber \\[4pt]
&= R V^{\dagger}  
  \frac{1}{\sqrt{\mathcal{N}}}\bigg[
    \sum_{x \in \Omega} \sum_{x' \in \Theta - \Omega} 
    \sqrt{e^{-\beta f(x)}} 
    \ket{x}_S \ket{f(x)}_F 
    \ket{x'}_{S'} \ket{f(x')}_{F'} 
    \ket{r_{x'}}_R 
    \ket{0}_C
  \bigg].
\end{align}

    Combining all the terms:
\begin{align}
\small
& R V^{\dagger} B^{\dagger} F B V \ket{\Pi}  
  \nonumber \\[4pt]
&= R V^{\dagger}  
  \frac{1}{\sqrt{\mathcal{N}}}\bigg[
    \sum_{x \in \Omega} \sum_{x' \in \Omega} 
    \sqrt{e^{-\beta f(x)}} 
    \ket{x}_S \ket{f(x)}_F 
    \ket{x'}_{S'} \ket{f(x')}_{F'} 
    \ket{m'}_R 
    \sqrt{1-A(x,x')} 
    \big[
      \sqrt{1-A(x,x')} \ket{0}_C 
      - \sqrt{A(x,x')} \ket{1}_C
    \big]
  \nonumber \\[-2pt]
&\qquad
  + \sum_{x' \in \Omega} \sum_{x \in \Omega} 
    \sqrt{e^{-\beta f(x)}} 
    \ket{x}_S \ket{f(x)}_F 
    \ket{x'}_{S'} \ket{f(x')}_{F'} 
    \ket{m'}_R 
    \sqrt{A(x,x')} 
    \big[
      \sqrt{A(x,x')} \ket{0}_C 
      + \sqrt{1-A(x,x')} \ket{1}_C
    \big]
  \nonumber \\[-2pt]
&\qquad
  + \sum_{x \in \Omega} \sum_{x' \in \Theta - \Omega} 
    \sqrt{e^{-\beta f(x)}} 
    \ket{x}_S \ket{f(x)}_F 
    \ket{x'}_{S'} \ket{f(x')}_{F'} 
    \ket{r_{x'}}_R 
    \ket{0}_C
  \bigg]
  \nonumber \\[4pt]
&= R V^{\dagger}  
  \frac{1}{\sqrt{\mathcal{N}}}\bigg[
    \sum_{x \in \Omega} \sum_{x' \in \Omega} 
    \sqrt{e^{-\beta f(x)}} 
    \ket{x}_S \ket{f(x)}_F 
    \ket{x'}_{S'} \ket{f(x')}_{F'} 
    \ket{m'}_R 
    \ket{0}_C   \nonumber \\[4pt]
&    + 
    \sum_{x \in \Omega} \sum_{x' \in \Theta - \Omega} 
    \sqrt{e^{-\beta f(x)}} 
    \ket{x}_S \ket{f(x)}_F 
    \ket{x'}_{S'} \ket{f(x')}_{F'} 
    \ket{r_{x'}}_R 
    \ket{0}_C
  \bigg]
  \nonumber \\[4pt]
&= R V^{\dagger}  
  \frac{1}{\sqrt{\mathcal{N}}}\bigg[
    \sum_{x \in \Omega} \sum_{x' \in \Theta} 
    \sqrt{e^{-\beta f(x)}} 
    \ket{x}_S \ket{f(x)}_F 
    \ket{x'}_{S'} \ket{f(x')}_{F'} 
    \ket{r_{x'}}_R 
    \ket{0}_C
  \bigg].
\end{align}

    And $V^{\dagger}$ returns it to state $\ket{\Pi}$. Applying $R$ does not change anything:
    \begin{align}
    \small
        & R V^{\dagger} \frac{1}{\sqrt{\mathcal{N}}} \sum_{x \in \Omega} \sum_{x' \in \Theta} \sqrt{e^{- \beta f(x)}} \ket{x}_S \ket{f(x)}_F \ket{x'}_{S'}\ket{f(x')}_{F'} \ket{r_{x'}}_R \ket{0}_C \nonumber \\ 
        & = R \frac{1}{\sqrt{Z(\beta)}} \sum_{x \in \Omega} \sqrt{e^{- \beta f(x)}} \ket{x}_S \ket{f(x)}_F \ket{0}_{S'}\ket{0}_{F'} \ket{0}_R \ket{0}_C \nonumber \\ 
        & = \ket{\Pi} 
    \end{align}
\end{proof}

It is important to note at this point that even though the exponential function has been approximated, the algorithm remains robust. To this end, let us formulate the following corollary:
\begin{corollary} \label{th:aprox_funct}
    Let the state be $\ket{\Pi_{p}} = \frac{1}{\sqrt{\mathcal{N}}} \sum_{x \in \Omega} \sqrt{p(x)} \ket{x}_S \ket{f(x)}_F \ket{0}_{S'}\ket{0}_{F'} \ket{0}_R  \ket{0}_C$, with $\mathcal{N} = \sum_{x \in \Omega} p(x)$, and let the apmlitude be $A(x,x')=\min \{ 1, \frac{p(x')}{p(x)}\}$ such that $\sqrt{p(x)}A(x,x')=\sqrt{p(x')}A(x',x)$. Then, $\ket{\Pi_{p}}$ is eigenstate of $W \;=\; R\,P^{\dagger}\,F\, P $ with eigenvalue 1. 
\end{corollary}

\begin{proof}
 The proof is identical to Theorem \ref{eigen_W}, and the key is that the balance equation is satisfied.
\end{proof}

The Corollary $\ref{th:aprox_funct}$ shows that the algorithm works regardless of the function we use, provided that we maintain the Metropolis-Hastings probability transition rule $A(x,x')=\min \{ 1, \frac{p(x')}{p(x)}\}$ so that the balance equation $\sqrt{p(x)}A(x,x')=\sqrt{p(x')}A(x',x)$ is satisfied. 

Thus, if we choose $p(x)\approx e^{\beta f(x)}$, it will be this approximation what the algorithm will sample, and if $p(x)$ is capable of separating the minima/maxima as the Gibbs distribution does, the algorithm will work correctly. Since we force the approximation to be a monotonic function, it will always be well--behaved.

Let us show that the quantum walk constructed can be expressed as a product of two reflections, and thus apply the following theorem, which is an adaptation of Theorem 1 in Ref. \cite{1366222}, as can be seen in Ref. \cite{PhysRevResearch.5.033059}:

\begin{theorem}
    Let $\mathcal{H}$ be a $N$-dimensional Hilbert space.
    Let $\mathcal{A}$ (resp. $\mathcal{B})$ be a $n$-dimensional subspace of $\mathcal{H}$ spanned by orthonormal vectors $u_1,...,u_m$ (resp. $v_1,...,v_n$).
    Denote by $V_{\mathcal{A}}$ (resp. $V_{\mathcal{B}}$) the $N \times m$ (resp. $N \times n$) matrix whose $i$th column is $u_i$ (resp. $v_i$).
    Define $R_{\mathcal{A}}=2V_{\mathcal{A}}V_{\mathcal{A}}^\dagger-I$ and $R_{\mathcal{B}}=2V_{\mathcal{B}}V_{\mathcal{B}}^\dagger-I$.
    Then, on $\mathcal{A}+\mathcal{B}$, the unitary operator $R_{\mathcal{A}}R_{\mathcal{B}}$ has an eigenvalue 1 with multiplicity 1, and any other eigenvalue is either of $e^{2i\theta_1},e^{-2i \theta_1},...,e^{2i\theta_l},e^{-2i\theta_l}$ or -1, where $\theta_1,...,\theta_l\in\left(0,\frac{\pi}{2}\right)$ are written as $\theta_i=\arccos \lambda_i$ with singular values $\{\lambda_i\}$ of $V_{\mathcal{A}}^\dagger V_{\mathcal{B}}$ that lie in $(0,1)$.
    \label{th:Szegedy}
\end{theorem}

Let us define the following subspace:
\begin{equation} \label{eq:subespA}
    \mathcal{A}:=\text{span} \{ \ket{x}_S \ket{f(x)}_F \ket{0}_{S'}\ket{0}_{F'} \ket{0}_R \ket{0}_C, x \in \Theta \},
\end{equation}
and 
\begin{equation}\label{eq:subespB}
    \mathcal{B}:=  P^{\dagger}F P \mathcal{A}.
\end{equation}

It is easy to verify that the restriction of $\Pi_0 P^\dagger FP \Pi_0$ to $\mathcal{A}$ is equal to $V_\mathcal{A}^\dagger V_\mathcal{B}$, where $\Pi_0$ is the projector onto $\mathcal{A}$, since the $(k,l)$ entry of the restriction of $\Pi_0 P^\dagger FP \Pi_0$ to $\mathcal{A}$ is:
    \begin{align}
                &\bra{x_k}_{S} \bra{f(x_k)}_F\bra{0}_{S', F', R, C}\Pi_0 P^\dagger FP \Pi_0 \ket{x_l}_{S} \ket{f(x_l)}_F\ket{0}_{S', F', R, C} \notag \\
        &= \bra{x_k}_{S} \bra{f(x_k)}_F\bra{0}_{S', F', R, C} P^\dagger FP \ket{x_l}_{S} \ket{f(x_l)}_F\ket{0}_{S', F', R, C}.
    \end{align}
wich from the definitions of $V_\mathcal{A}$ and $V_\mathcal{B}$ are the same $(k,l)$ entry of $V_\mathcal{A}^\dagger V_\mathcal{B}$.

This allows us to arrive at the following result:
\begin{theorem}
    Consider the Markov chain generated by $W=RP^\dagger FP$ and denote by $\Delta$ its spectral gap. Then, on the subespace $\mathcal{A}+\mathcal{B}$ defined in (\ref{eq:subespA}) and (\ref{eq:subespB}), $\ket{\Pi}$ is the unique eigenstate of $W$ with eigenvalue 1, and any other eigenvalue is written as $e^{i\theta}$ with $\theta\in\mathbb{R}$ such that $|\theta|\ge\arccos(1-\Delta)$.
    \label{th:uniquenessEigenstate}
\end{theorem}

\begin{proof} \label{eq:RA}
    Let us define the following reflections:
    \begin{equation}
        R_\mathcal{A} = 2\sum_{x\in \Theta} \ket{x}_{S} \ket{f(x)}_F\ket{0}_{S', F', R, C}\bra{x}_{S} \bra{f(x)}_F\bra{0}_{S', F', R, C}-I
    \end{equation}
    and
    \begin{align} \label{eq:RB}
        R_\mathcal{B}
        &= 2\sum_{x\in\Theta}  P^\dagger FP \ket{x}_{R_{\rm S}} \ket{x}_{S} \ket{f(x)}_F\ket{0}_{S', F', R, C}\bra{x}_{S} \bra{f(x)}_F\bra{0}_{S', F', R, C} (P^\dagger FP )^\dagger - I \nonumber \\
        &= P^\dagger FP R_\mathcal{A} P^\dagger FP.
    \end{align}
    Thus, if we define the operator given by these two reflections,
    \begin{equation}
        U=R_\mathcal{A}R_\mathcal{B} = R_\mathcal{A} P^\dagger FP R_\mathcal{A} P^\dagger FP = W^2,
    \end{equation}
    where the reflection $R$ of the construction $W=RP^\dagger FP$ acts as $R_\mathcal{A}$ on $\mathcal{A}+ \mathcal{B}$.
    Therefore, on $\mathcal{A}+ \mathcal{B}$, the eigenvalues of $W$ are equal to the square root oh those of $R_\mathcal{A}R_\mathcal{B}$. By Theorem \ref{th:Szegedy}, they include 1 or -1 with multiplicity 1, and, because of Theorem \ref{eigen_W}, it is in fact 1 with the corresponding eigenstate $\ket{\Pi}$.
    Any other eigenvalue of $U_W$ is $e^{\pm i\theta_l}$, $-e^{\pm i\theta_l}=e^{i(\pm\theta_l+\pi)}$, or $\pm i=e^{\pm \frac{\pi}{2}i}$, whose phase has modulus no less than
    \begin{equation}
    \arccos \left(\max \{|\lambda_l|\}\right)=\arccos (1-\Delta)
    \end{equation}
    in any case. 
\end{proof}

Theorem~\ref{th:uniquenessEigenstate} guarantees that the state $\ket{\Pi_{\beta}}$ is the unique eigenstate of the operator with eigenvalue $1$. This property is crucial for the correct operation of the Quantum Simulated Annealing (QSA) algorithm. The existence of a single, well-defined stationary eigenstate ensures that the application of quantum phase estimation will unambiguously project the system onto $\ket{\Pi_{\beta}}$ when measuring the eigenvalue register in the $\ket{0}^{\otimes p}$ state. Consequently, the subsequent randomization procedure~\cite{PhysRevLett.101.130504} faithfully reproduces the desired sampling dynamics, allowing the algorithm to converge toward the target distribution encoded in $\ket{\Pi_{\beta}}$.

It is important to clarify the role of the initial state preparation. 
Since no a priori structural information about the feasible region is assumed, the register $S$ is initialized in the uniform superposition over all computational basis states by applying Hadamard gates. 
As a consequence, the initial state necessarily contains components corresponding to configurations that violate one or more constraints. 
This does not pose a problem for the algorithm: infeasible configurations are not reinforced by the dynamics and cannot accumulate additional probability mass, since new candidate states are generated coherently in the auxiliary register $S'$ and accepted only conditionally. 
Moreover, the intermediate partial measurements applied to the ancillary registers progressively eliminate entanglement with rejected or invalid transitions, causing the amplitude associated with infeasible states to decay. 
As the walk proceeds, probability flows toward the feasible subspace, while the residual weight outside it vanishes asymptotically. 
Thus, despite starting from a fully unstructured superposition, the combined effect of the Metropolis filtering and the intermediate measurements ensures convergence toward the physically relevant subspace without requiring an explicit feasibility projection at initialization.

\end{document}